\documentclass[11pt, one column]{article}

\usepackage[english]{babel}
\usepackage[utf8x]{inputenc}
\usepackage[T1]{fontenc}

\usepackage[margin=1in]{geometry}

\usepackage{graphicx} 
\usepackage[colorlinks=False]{hyperref} 
\usepackage{amsmath}  
\usepackage{amsfonts} 
\usepackage{amssymb}  
\usepackage{amsthm}
\usepackage[shortlabels]{enumitem}
\usepackage{todonotes}
\usepackage{multirow}
\usepackage{float}
\usepackage{pgfplots}
\usepackage{tikz}
\usetikzlibrary{calc}
\usepackage{authblk}
\usepackage{cite}

\newcommand{\C}{\mathcal{C}}

\newcommand{\F}{\mathcal{F}}

\newcommand{\cm}{\textsc{Max-}f}
\newcommand{\m}{M}
\newcommand{\optm}{O_{\m}}
\newcommand{\am}{\alpha_{\textsc{Max-}f}}

\newcommand{\cs}{\textsc{Sum-}f}
\newcommand{\s}{\scriptscriptstyle{\Sigma}}
\newcommand{\opts}{O_{\s}}
\newcommand{\as}{\alpha_{\textsc{Sum-}f}}

\newcommand{\cmm}{\textsc{Max-Max}}
\newcommand{\mm}{\scriptscriptstyle{MM}}
\newcommand{\optmm}{O_{\mm}}
\newcommand{\amm}{\alpha_{\mm}}

\newcommand{\cms}{\textsc{Max-Sum}}
\newcommand{\ms}{\scriptscriptstyle{M \Sigma}}
\newcommand{\optms}{O_{\ms}}
\newcommand{\ams}{\alpha_{\ms}}

\newcommand{\csm}{\textsc{Sum-Max}}
\newcommand{\sm}{\scriptscriptstyle{\Sigma M}}
\newcommand{\optsm}{O_{\sm}}
\newcommand{\asm}{\alpha_{\sm}}

\newcommand{\css}{\textsc{Sum-Sum}}
\newcommand{\sums}{\scriptscriptstyle{\Sigma \Sigma}}
\newcommand{\optss}{O_{\sums}}
\newcommand{\ass}{\alpha_{\sums}}

\newtheorem{theorem}{Theorem}[section]
\newtheorem{corollary}{Corollary}[theorem]
\newtheorem{lemma}[theorem]{Lemma}

\newtheorem{definition}{Definition}[section]

\title{Compatibility of Max and Sum Objectives \\ for Committee Selection and $k$-Facility Location}

\author{Yue Han}
\author{Elliot Anshelevich}
\affil{Department of Computer Science, Rensselaer Polytechnic Institute}
\affil{hany4@rpi.edu, eanshel@cs.rpi.edu}

\date{\today}
\begin{document}
\maketitle

\begin{abstract}
    We study {a version of} the metric facility location problem (or, equivalently, {variants of} the committee selection problem) in which we must choose $k$ facilities in an arbitrary metric space to serve some set of clients $\mathcal{C}$. We consider four different objectives, where each client $i\in \mathcal{C}$ attempts to minimize either the sum or the maximum of its distance to the chosen facilities, and where the overall objective either considers the sum or the maximum of the individual client costs.
    Rather than optimizing a single objective at a time, we study how {\em compatible} these objectives are with each other, and show the existence of solutions which are simultaneously close-to-optimum for any pair of the above objectives. Our results show that when choosing a set of facilities or a representative committee, it is often possible to form a solution which is good for several objectives at the same time, instead of sacrificing one desideratum to achieve another. 
\end{abstract}

\section{Introduction}
\label{section:intro}
Metric facility location problems {and their variants}
form a classic and well-known area of research. In these problems, we are given a set of possible facility locations $\F$ and a set of $n$ clients $\C$ in an arbitrary metric space $(M,d)$. {In the variants that we study in our work,} the goal is to choose a set of $k > 1$ facilities $A$ that optimizes some objective, usually making sure that the distance between clients and facilities is not too large. 
Such problems have been studied extensively in various areas such as operations research (see survey \cite{farahani2010multiple}), approximation algorithms (see book \cite{williamson2011design} and surveys \cite{davidnotes,an2017recent}),
and mechanism design (see survey \cite{chan2021mechanism}) for decades.
Metric facility location {and its variants are} general enough to capture many important settings. For example, consider when a city needs to choose a location to build a hospital from a set of potential locations, or a neighborhood wishes to choose multiple locations, each dedicated for one purpose (e.g. school, post office, library, grocery store).
Facility location can also be thought of as a committee selection problem in a spacial voting setting \cite{merrill1999unified, enelow1984spatial}, where each candidate is a potential facility location and each voter is a client; the goal is to choose $k$ candidates to form a committee which is somehow representative of all the voters according to some objective.

The objective being optimized, i.e., the measure of what makes a committee or a set of facilities be ``good'', is a crucial component of facility location problems {and their many versions}. Many different objectives have been studied in the past (see Related Work). It is not clear, however, what the correct objective is for most settings. 
For example, should we care about utilitarian objectives (minimizing the average client cost), or egalitarian objectives (making sure that all clients are not too unhappy)? 
What determines if a committee or a set of locations is good from a client's perspective? Optimizing a single chosen objective can often make the cost of a solution in terms of other objectives be extremely bad, even if these other objectives are just as reasonable as the one we decided to optimize. In our work, instead of selecting a single specific objective, we would like to see if we are capable of finding a solution that would be {\em simultaneously good for multiple objectives}. 

To see what kind of objectives we will optimize, first consider the individual cost for each client. Here note that instead of assigning each client to one facility, we are interested in modeling settings and applications where clients need to utilize all of the $k$ facilities. For example, as mentioned above, consider the case where a neighborhood wants to build $k$ different facilities to fulfill the needs of the residents in that area, e.g., a shopping center, a post office, and a grocery store. The residents will use all of these, but should they care more about not living too far away from any of these facilities, or about having a short average distance from their home to these facilities? In other words, is the goal to {\it minimize the maximum distance} from the client to all facilities ({\it max-variant} \cite{chen2020facility, lotfi2024truthful, zhao2023constrained}), or is the goal to {\it minimize the average (or total) distance} from the client to all facilities ({\it sum-variant} \cite{kanellopoulos2023truthful,lotfi2024truthful, zhao2024constrained})? Similar objectives arise in spacial voting settings as well, where distance between a voter and a committee member represents ideological differences. 
Does a voter care about {\em all} of the committee members being ideologically similar to them ({\it max-variant}), or about the {\em total} ideological difference between them and the committee members ({\it sum-variant})? While many other important objectives exist{, such as the classical setting where each client is assigned to their {\it closest} facility}, we focus on the max and sum variants in our work, and attempt to optimize {\em both} of them simultaneously.

After fixing the cost for each client, we also have many choices on how to calculate the overall cost of a committee or a set of facilities. Should we care about keeping the average of the individual costs as low as possible (a utilitarian measure), or should we keep the maximum cost of every individual low (an egalitarian measure)? With all the above individual costs, as well as different ways of combining them into an overall objective, this results in numerous possible objectives, with none necessarily ``better'' than the other. Because of this, we focus on optimizing multiple objectives simultaneously. While many other reasonable objectives exist, in our work we focus on the following four natural objectives:
\begin{definition} Let $A$ be a set of $k$ facilities and $\C$ be the set of clients in metric space $(M,d)$. We define the following:
\begin{itemize} 
	\item \cmm(A) = $\max_{i\in \C} \max_{a\in A} d(i,a)$
	\item \cms(A) = $\max_{i\in \C} \sum_{a\in A} d(i,a)$
	\item \csm(A) = $\sum_{i\in \C} \max_{a\in A} d(i,a)$
	\item \css(A) = $\sum_{i\in \C} \sum_{a\in A} d(i,a)$
\end{itemize}
	\label{def:objs}
\end{definition}


Our goal in this paper is to form solutions which are close-to-optimal for more than just a single objective. Formally, we want to form solutions which simultaneously approximate two objectives at the same time.

\begin{definition}
	{Simultaneous Approximation}: Let $c_1$ and $c_2$ be two different objectives, and let $\mathcal{A}$ be the set of all possible solutions. Let $O_i$ be an optimum solution for objective $c_i$, i.e., $O_i= \arg\min_{A\in \mathcal{A}} c_i(A)$. We then define the approximation ratio of solution $A$ with respect to objective $c_i$ as
	$$\alpha_{c_i}(A) = \frac{c_i(A)}{c_i(O_i)} \geq 1.$$
	Therefore, by choosing $A$, we would obtain a $(\alpha_{c_1}(A), \alpha_{c_2}(A))$ approximation for minimizing the two objectives. If we let $\alpha = \max\{\alpha_{c_1}(A), \alpha_{c_2}(A)\}$, this means that $A$ is within a factor $\alpha$ for minimizing {\em both} of the objectives. Hence, we define $\alpha$ as the simultaneous approximation ratio of $A$ for objectives $c_1$ and $c_2$.
	\label{def:sa}
\end{definition}

In other words, a solution $A$ which simultaneously approximates two objectives within a factor $\alpha$ is simply a solution which is within a factor $\alpha$ of optimum for each objective, and thus a solution which is an $\alpha$-approximation for each objective individually. 



In our work, we call a pair of objectives {\em $\alpha$-compatible} or just {\em compatible} if there always exists a solution which simultaneously approximates both objectives within some small constant $\alpha$. This shows that it is possible to (approximately) optimize both objectives at once, and we do not have to sacrifice one objective in order to form a good solution for the other. Our goal, then, is to see if the four objectives listed above are compatible with each other, as well as quantify the upper and lower bounds of the simultaneous approximation ratio for each pair of them. 

\subsection{Our Contributions} 
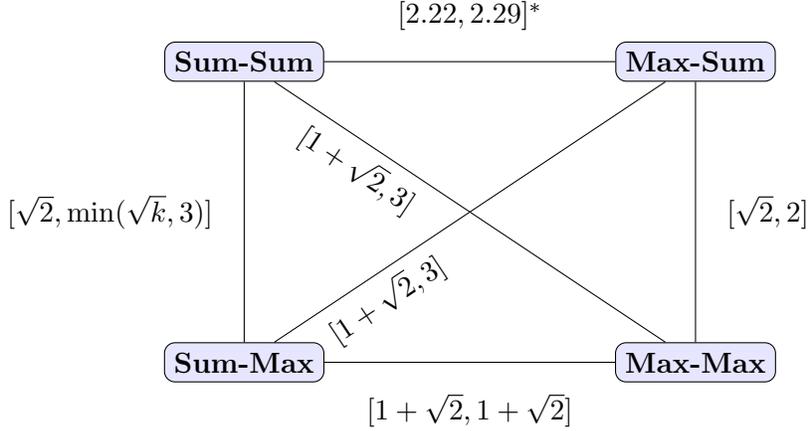
\begin{figure}[t]
  \begin{center}
   \begin{tikzpicture}
        \node[draw, fill=blue!10, rounded corners] (SS) at (-3,2) {\textbf{Sum-Sum}};
        \node[draw, fill=blue!10, rounded corners] (MS) at (3,2) {\textbf{Max-Sum}};
        \node[draw, fill=blue!10, rounded corners] (SM) at (-3,-2) {\textbf{Sum-Max}};
        \node[draw, fill=blue!10, rounded corners] (MM) at (3,-2) {\textbf{Max-Max}};
        
        \draw (SS) -- (MS) node[midway, above=8pt] {$[2.22, 2.29]^*$};
        \draw (SS) -- (SM) node[midway, left=8pt] {$[\sqrt{2}, \min(\sqrt{k},3)]$};
        \draw (MS) -- (MM) node[midway, right=8pt] {$[\sqrt{2},2]$};
        \draw (SM) -- (MM) node[midway, below=8pt] {$[1+\sqrt{2},1+\sqrt{2}]$};
        \draw (SS) -- (MM) node[midway, sloped, pos=0.25,below] {$[1+\sqrt{2},3]$};
        \draw (SM) -- (MS) node[midway, sloped, pos=0.25,below] {$[1+\sqrt{2},3]$};
    \end{tikzpicture}
        
            
            
  \caption{A summary of our results. Each line connecting two objectives includes a label showing how compatible we prove these two objectives to be. Specifically, if a line has a label $[x,y]$, this means we show there always exists a solution which is a $y$-approximation to the optimum for both of these objectives simultaneously, as well as a lower bound showing no simultaneous approximation better than $x$ is possible. Note that $(*)$ only holds when $k \geq 3$ and is an approximation for the actual bounds, $[(4 + \sqrt{7})/3\approx 2.22,1 + \sqrt{5/3}\approx 2.29]$.}
  \label{fig:results}
  \end{center}
\end{figure}

We begin by proving that for any pair of our objectives, there always exists a solution which is simultaneously a 3-approximation for both objectives. Moreover, this solution is simply the optimum solution for \css. This simple result shows that all of the objectives are 3-compatible. We then proceed to improve the 3-simultaneous approximation ratio for individual pairs of objectives; our results are summarized in Figure \ref{fig:results}. 

To do this, we first show that results from \cite{han2022optimizing} which apply to placing a single facility (i.e., $k=1$), extend without much difficulty to simultaneously approximating the pair $(\css, \cms)$, as well as the pair $(\csm, \cmm)$. This immediately implies that both pairs of objectives are $(1+\sqrt{2})$-compatible.
We then proceed with further improvements to these bounds, as shown in our main contributions:

\begin{description}
\item[\css ~and \cms] While the above approximation bound of $1+\sqrt{2}\approx 2.41$ is tight for \csm ~and \cmm, we can show better bounds for simultaneously approximating \css ~and \cms ~when choosing more than 2 facilities. In fact, in Section \ref{sec:ssms} we prove that these objectives are simultaneously approximable to within a factor of $1 + \sqrt{\frac{5}{3}} \approx 2.29$, as long as $k\geq 3$. This may be somewhat surprising, as usually things get worse and more complex as the number of facilities $k$ becomes larger. For the questions we are asking, however, the reverse turns out to be true: when $k$ is small, then only a few possible solutions may exist, and they may all be bad for at least one of the two objectives. When $k$ is large, however, many solutions become possible, and we are able to always form a solution which is good for {\em both} objectives by carefully stitching together parts of the optimum solutions for each separate objective. The solution we form is not optimum for either, but is a good approximation for both. 
\item[\cms ~and \cmm]
We then proceed to consider {\cms} and {\cmm}, {\css} and {\csm}, for which the results from \cite{han2022optimizing} can no longer be extended directly; the compatibility of these objectives has not been considered before. While we show both of these pairs of objectives are simultaneously approximable to within a factor of $\min(\sqrt{k},3)$, we are able to significantly improve this bound for the specific case of \cms ~and \cmm. We prove that these objectives are 2-compatible, showing that if we care about both of these costs, we are able to obtain a solution that is close-to-optimum for both. This result requires somewhat different techniques, and is presented in Section \ref{section:mmms}. 
\end{description}

\subsection{Related Work}

Many variants of facility location problems, as well as spacial voting problems, have been studied within many disciplines and are too much to survey here; see surveys \cite{farahani2010multiple, chan2021mechanism}. Single-facility location (i.e., $k=1$) has been especially well-studied, while building multiple facilities has received somewhat less attention to due to the added complexity. When choosing the locations for multiple facilities, the most commonly considered cost for each client is the distance from the client to their closest facility \cite{sornat2019approximation, chakrabarty2017interpolating,alamdari2017bicriteria, kumar2006fairness, tamir2001k, CHAN2023114208}. 
In our work, however, we are interested in settings where the client cares about the locations of all $k$ facilities, not just the closest one. These settings arise when the facilities being built are heterogeneous \cite{serafino2016heterogeneous} (for example one is a post office, another is a grocery store, etc, so the client has to use them all). In social choice settings, this corresponds to the idea that the voter cares about the entire committee membership, not just the member which is most similar to them. 
Thus, we instead consider other types of client costs, specifically the maximum distance to a built facility (as in \cite{chen2020facility, lotfi2024truthful, zhao2023constrained}), and the total (or average) distance to the built facilities (as in \cite{kanellopoulos2023truthful,lotfi2024truthful, zhao2024constrained,xu2021two}). 
Although we focus on the above two cost types, there are certainly other client cost functions that have been studied and are interesting. 
For example, \cite{caragiannis2022metric} studies $q$-$social$ $cost$, where the cost of choosing a committee of size $k$ for each voter is the distance from them to their $q$'th closest alternative in the committee, and \cite{gai2024two} considers the scenario where different 
clients could have different cost functions. 
For the overall objective combining the costs of the clients, max and sum are commonly used (see for example, \cite{chan2021mechanism, walsh2021strategy, CHAN2023114208,chen2020facility, lotfi2024truthful, kanellopoulos2023truthful}), but are rarely considered together under the notion of simultaneous approximation.  Note that while computing the optimum solution for some of these objectives is NP-Complete, in this work we are more focused on which objectives are compatible with each other in principle, leaving the question of poly-time simultaneous approximation for future work.

A significant amount of work also exists on optimizing several objectives at once for facility location. 
However, as described in \cite{ehrgott2000survey, farahani2010multiple}, the most commonly studied way to do this is to convert multiple objectives into a single objective. For example, \cite{mcginnis1978single, ohsawa1999geometrical} combine two objectives together using a convex combination of them. Moreover, \cite{alamdari2017bicriteria} considers a slightly different measure that would achieve a $(4,8)$ approximation for $k$-$Center$ and $k$-$Median$ problems w.r.t the optimal solution of a convex combination of the two objectives. 

Instead of such previous approaches, we attempt to approximate a pair of objectives {\em simultaneously}, so the solution formed is close-to-optimal for each of them {\em at the same time.} This notion of approximation has received far less attention (see the Related Work section in \cite{han2022optimizing} for a detailed discussion). While \cite{han2022optimizing} studied exactly this notion for exactly our setting, they only considered building one facility ($k=1$); in this paper we greatly generalize their results to $k>1$ facilities, as well as use new techniques for multiple facilities to form better approximation bounds. Although not the main focus of their work, \cite{alamdari2017bicriteria} considered this notion of approximation as well, and showed that if each client's cost is their distance to their closest facility, the simultaneous approximation ratio even for choosing two facilities for minimax and minisum can be arbitrarily large. While our results establish that many sum and max objectives are simultaneously compatible, the results from \cite{alamdari2017bicriteria} show that this is not true if we care about the distance to the closest facility instead. 
There is also other previous work, such as  \cite{kumar2006fairness, gkatzelis2020resolving}, which either considers the simultaneous approximation of max and sum objectives in our setting, or implies results that would fit under this model. However, such work only considers choosing a single facility, while our results hold for an arbitrary number of facilities $k$.



Another nice line of literature which studies facility location and committee selection exists in the area of mechanism design (see survey \cite{chan2021mechanism}). Such literature is mostly interested in finding mechanisms such that no clients have the incentive to lie about their location or preferences, i.e., a truthful (or, strategy-proof) mechanism.
In addition, much of this work is focused on choosing a single facility \cite{CHAN2023114208,chen2020facility, kanellopoulos2023truthful}, or two facilities \cite{deligkas2024agent,lotfi2024truthful,xu2021two}, and often only on a line, instead of in an arbitrary metric space. \cite{walsh2021strategy} provides a 3-simultaneous approximation deterministic strategy-proof mechanism for minimax and minisum for choosing one facility on a line with limited location options, but also shows that no deterministic strategy-proof mechanism can do better than 3.  The setting from \cite{lotfi2024truthful} may be the closest to ours in this area, which considers the two-facility location problem ($k=2$) on a line and provides a (5,11) truthful deterministic mechanism for {\cmm} and {\cms}. Furthermore, 5 and 11 are tight bounds for the truthful approximation ratio for {\cmm} and {\cms} respectively. In addition, they also showed a 22 simultaneous approximation truthful deterministic mechanism for all of the four objectives from Definition \ref{def:objs}. In this paper, we study the case where each client cares about all of the facilities in a general metric space without the consideration of strategy-proofness; our goal is to quantify what is possible to achieve when caring about multiple objectives, and which objectives are compatible with each other.

\section{Extending Previous Results to Multiple Facilities ($k>1$)}
\label{section:prelim}
Consider {the version of} the facility location problem where we are given the set $\C$ of $n$ client locations, and the multiset of possible facility locations $\F$ \footnote{Here we note that $\F$ is a multiset such that the multiplicity of a facility location in $\F$ is the number of times that location can be used/chosen.} in a metric space $(M,d)$. We want to place $k$ facilities so that the placement would simultaneously be close to optimal for several of our objectives. Denote the set of $k$ locations by $A$ such that $A \subseteq \F, |A|=k$. The set of objectives that we are interested in are chosen from objectives defined in Definition \ref{def:objs}. In this section, however, it is convenient to think more generally, and consider an arbitrary function $f:M \times M^{k} \rightarrow \mathbb{R^{+}}\cup \{0\}$ so that $f(i,A)$ denotes the cost of client $i$ for solution $A$. For such a fixed function $f$, we can also define the following.


\begin{definition}
	Let $A \subseteq \F, |A| = k$. We define {\cm} and {\cs} as follows:
	\begin{itemize}
		\item {\cm} = $\max_{i \in \C}(f(i,A))$.
		\item {\cs} = $\sum_{i \in \C}(f(i,A))$.
	\end{itemize}
	Where $f:M \times M^{k} \rightarrow \mathbb{R^{+}}\cup \{0\}$.
	\label{def:maxsum}
\end{definition}

In existing work, \cite{han2022optimizing} considered the case where $k=1$, with the cost function being \textsc{Max} and \textsc{Sum} where $f(i, A) = d(i, a), A = \{a\} \subseteq \F, i \in \C$.  Consider, however, a more general setting with $k>1$, and suppose that the function $f$ obeys the following inequality (which is essentially just the triangle inequality):
\begin{equation}
	f(i,A) \leq f(i,B) + f(j,B) + f(j,A) \qquad \forall i,j\in \C, A \subseteq \F, B \subseteq \F
	\label{equ:f}
\end{equation}
{Note that some classic objectives such as $k$-median, for which $f$ is the $\min$ function, {\em do not} obey this property. As we discuss in Section \ref{sec:examples_f}, however, many important objectives do satisfy this property.}
If $f$ satisfies this inequality, consider a new metric $(M', d')$ where each client in $\C$ and each subset of $k$ facilities in $\F$ are a point in $M'$. We denote the set of client points by $\C'$ and the set of facility combinations by $\F'$. Then, we can define a new metric $d'$ as
\begin{enumerate}
	\item $d'(i,j) = d(i,j), i,j \in \C'$ 
	\item $d'(i,A) = d'(A,i) = f(i,A), i \in \C', A \in \F'$
	\item $d'(A,B) = \min_{i\in \C'}(f(i,A)+f(i,B)),A,B \in \F'$
\end{enumerate}
Since $(M,d)$ is a metric, combined with Inequality \ref{equ:f}, we can see that $d'$ also obeys triangle inequality. Therefore, we can reduce the problem from choosing $k$ facilities in metric $(M,d)$ with clients $\C$ and facility locations $\F$ to choosing 1 facility in metric $(M',d')$ with clients $\C'$ and facility locations $\F'$.
Due to this reduction, results from \cite{han2022optimizing} immediately extend to choosing multiple facilities.\footnote{\cite{han2022optimizing} actually considers more than just $\cs$ and $\cm$. The extension we showed above also holds for all the other cost functions over all clients considered in \cite{han2022optimizing} (namely the {\it l-centrum} objectives, see Definition \ref{def:cent}).} In particular, the following results still hold:

\begin{theorem} {\bf (generalization of Theorem 3.1 from \cite{han2022optimizing})}
Given the optimal facility location set $\opts$ that minimizes $\cs$ and the optimal facility location set $\optm$ that minimizes $\cm$, we have that $\am(\opts) \leq \frac{1}{\as(\optm)} + 2$ as long as $f$ obeys Inequality \ref{equ:f}.
\label{theorem:2_gen}
\end{theorem}

\begin{corollary} {\bf (generalization of Corollary 3.1.1 from \cite{han2022optimizing})} When $f$ obeys Inequality \ref{equ:f}, we have that
\begin{enumerate}
	\item By choosing the optimal facility location set $\opts$ that minimizes $\cs$, we obtain a $(3, 1)$ approximation for simultaneously minimizing $\cm$ and $\cs$. 
	\item There always exists a facility location set $A\subseteq\F$ such that choosing $A$ would give a $1+\sqrt{2}$ approximation both for minimizing $\cm$ and minimizing $\cs$. In fact, we would either get a $(1, 1+\sqrt{2})$ approximation by choosing $\optm$ or a $(1+\sqrt{2}, 1)$ approximation by choosing $\opts$. In other words, at least one of $\am(\opts)$ or $\as(\optm)$ is always less than or equal to $1+\sqrt{2}$. 
\end{enumerate}
\label{coro:bound_2}
\end{corollary}

\subsection{Examples of $f$}\label{sec:examples_f}
We will now look at some functions $f$ that obey Inequality \ref{equ:f}. Some examples of $f$ include using the centroid and any norms of the set of chosen facilities, since these kinds of functions would map each possible facility set/committee (set of $k$ facilities) into a new metric space, which means that $f$ would obey Inequality \ref{equ:f}. In addition, \cite{caragiannis2022metric} showed that the $q$-$social$ $cost$ of a committee of size $k$ with $q > k/2$ also obeys Inequality\ref{equ:f}. Note that for the four objectives that we consider ({\cms}, {\css}, {\cmm}, and {\csm}), we use the summation function and the maximum function as $f$. In fact, these two functions also obey Inequality \ref{equ:f}, as we prove below, immediately implying that both ({\cms}, {\css}) and ({\cmm}, {\csm}) are $(1+\sqrt{2})$-compatible. 

\begin{lemma}
	Let 
	$$f(i,A) = \sum_{a \in A}d(i,a) \qquad \forall i \in \C, A \subseteq \F,$$
	then we have that 
	$$f(i,A) \leq f(i,B) + f(j,B) + f(j,A) \qquad \forall i,j\in \C, A \subseteq \F, B \subseteq \F, |A| = |B| = k.$$
	\begin{proof}
		For clarity, let $A = \{a_1, a_2, \cdots, a_k\} \subseteq \F$, $B = \{b_1, b_2, \cdots, b_k\} \subseteq \F$, $i,j \in \C$, we then want to show that $f$ obeys the above inequality. Note that all clients and facility locations are in some metric space $(M,d)$, using triangle inequality we have  
\[
\begin{aligned}
	f(i,A) &= d(i, a_1) + d(i, a_2) + \cdots + d(i, a_k)\\
	&\leq d(i, b_1) + d(i, b_2) + \cdots + d(i, b_k) + d(a_1, b_1) + d(a_2, b_2) + \cdots + d(a_k, b_k)\\
	&\leq f(i, B) + d(j, a_1) + d(j, a_2) + \cdots + d(j, a_k) + d(j, b_1) + d(j, b_2) + \cdots + d(j, b_k)\\
	&= f(i,B) + f(j,B) + f(j, A).
\end{aligned}
\]
Since the choice of $A,B,i,j$ are arbitrary, we can conclude that $f$ indeed satisfies the statement. 
	\end{proof}
\label{lemma:sumf}
\end{lemma}

\begin{lemma}
	Let 
	$$f(i,A) = \max_{a \in A}d(i,a) \qquad \forall i \in \C, A \subseteq \F,$$
	then we have that 
	$$f(i,A) \leq f(i,B) + f(j,B) + f(j,A) \qquad \forall i,j\in \C, A \subseteq \F, B \subseteq \F.$$
	\begin{proof}
		Let $a^* \in A \subseteq \F, b^* \in B \subseteq \F, i,j \in \C$ such that $d(i,a^*) = \max_{a\in A}d(i,a),$ $ d(i,b^*) = \max_{b\in B}d(i,b)$. Note that we also have that $d(j, b^*) \leq \max_{b\in B}d(j,b),$ $ d(j, a^*) \leq \max_{a\in A}d(j,a)$. Then, by triangle inequality we have
\[
\begin{aligned}
	f(i,A) &= d(i,a^*)\\
	&\leq d(i, b^*) + d(j, b^*) + d(j, a^*)\\
	&\leq \max_{b\in B}d(i,b) + \max_{b\in B}d(j,b) + \max_{a\in A}d(j,a)\\
	&= f(i,B) + f(j,B) + f(j,A).
\end{aligned}
\]
Since the choice of $A,B,i,j$ are arbitrary, we can conclude that $f$ indeed satisfies the statement. 
	\end{proof}
\label{lemma:maxf}
\end{lemma}

Unfortunately, if the cost for each client $i \in \C$ is the distance from them to the closest facility in the chosen set $A \subseteq \F$, 
$$f(i,A) = \min_{a \in A}d(i,a),$$
then such cost functions do not obey Inequality \ref{equ:f}. In addition, as we have discussed previously, the simultaneous approximation ratio for $\textsc{Max-Min}$ and $\textsc{Sum-Min}$ is unbounded.

\section{3-Compatibility of our objectives}
We are interested in simultaneously approximating all pairs of the objectives defined in Definition \ref{def:objs}. Before analyzing each pair more carefully, we first make the observation that the optimal solution of \textsc{Sum-Sum} is a 3 approximation for all the four objectives as defined in Definition \ref{def:objs}. To show this, we will first consider \textsc{Sum-Sum} and \textsc{Max-Sum}. In fact, by Lemma \ref{lemma:sumf} we note that this is a special case of Corollary \ref{coro:bound_2} and can obtain the following result.

\begin{corollary}
	The optimal solution of \textsc{Sum-Sum} is a 3 approximation for \textsc{Max-Sum}.
	\label{coro:3ssms}
\end{corollary}

To continue on to the other two objectives, we will first show a useful lemma. Let $O$ be the optimal solution for \textsc{Sum-Sum}, $A$ be the optimal solution for \textsc{Sum-Max} and $B$ be the optimal solution for \textsc{Max-Max}. In addition, let $o^* = \operatorname{argmax}_{o \in O} \sum_{i \in \C}d(i, o)$, $a^* = \operatorname{argmax}_{a \in A} \sum_{i \in \C}d(i, a)$, $b^* = \operatorname{argmax}_{b \in B} \sum_{i \in \C}d(i, b)$. For each $i \in \C$, let $o_i = \operatorname{argmax}_{o\in O}d(i,o)$, $a_i = \operatorname{argmax}_{a\in A}d(i,a)$, $b_i = \operatorname{argmax}_{b\in B}d(i,b)$.

\begin{lemma}
	For any $P \subseteq \F, |P| = k, p^* = \operatorname{argmax}_{p \in P} \sum_{i \in \C}d(i, p)$, we have that $\sum_{i \in \C}d(i, o^*) \leq \sum_{i \in \C}d(i, p^*)$.
	\begin{proof}
		First note that if $P = O$, then it is trivial that $\sum_{i \in \C}d(i, o^*) \leq \sum_{i \in \C}d(i, p^*)$. Therefore, we assume $P \neq O$. We will prove the above claim using contradiction. Assume otherwise, $\sum_{i \in \C}d(i, o^*) > \sum_{i \in \C}d(i, p^*)$. Note that since $P\neq O, |P| = |O| = k$, by pigeonhole principle there must exist some $p \in P$ such that $p \notin O$. Recall that $p^* = \operatorname{argmax}_{p \in P} \sum_{i \in \C}d(i, p)$, this means that $\sum_{i \in \C}d(i, o^*) > \sum_{i \in \C}d(i, p^*) \geq \sum_{i \in \C}d(i, p)$. Now, let $O' = \{p\}\cup O \setminus\{o^*\}$, we have that $\sum_{i\in \C} \sum_{o'\in O'} d(i,o') = \sum_{i\in \C} \sum_{o\in O} d(i,o) - \sum_{i \in \C}d(i, o^*) + \sum_{i \in \C}d(i, p) < \sum_{i\in \C} \sum_{o\in O} d(i,o)$, which is a contradiction. Therefore, we can conclude that $\sum_{i \in \C}d(i, o^*) \leq \sum_{i \in \C}d(i, p^*)$.
	\end{proof}
	\label{lemma:sumsum}
\end{lemma}

Then, we can show the following results.

\begin{theorem}
	The optimal solution of \textsc{Sum-Sum} is a 3 approximation for \textsc{Sum-Max}.
	\begin{proof}
		By triangle inequality, we have that 
		\[
		\begin{aligned}
			\textsc{Sum-Max}(O) &= \sum_{i\in \C} d(i, o_i)\\
			&\leq \sum_i \left( d(i, a_i) + d(a_i, o_i)\right)\\
			&\leq \textsc{Sum-Max}(A) + \sum_i \frac{1}{n} \sum_j d(j, a_i) + \sum_i \frac{1}{n} \sum_j d(j, o_i)\\
			&\leq \textsc{Sum-Max}(A) + \sum_i \frac{1}{n} \sum_j d(j, a^*) + \sum_i \frac{1}{n} \sum_j d(j, o^*)\\
		\end{aligned}
		\]
		The last inequality is true due to how $a^*$ and $o^*$ are defined. Now, by Lemma \ref{lemma:sumsum}, we have $\sum_{i \in \C}d(i, o^*) \leq \sum_{i \in \C}d(i, a^*)$, therefore, 
		\[
		\begin{aligned}
			\textsc{Sum-Max}(O) &\leq \textsc{Sum-Max}(A) + \sum_i \frac{1}{n} \sum_j d(j, a^*) + \sum_i \frac{1}{n} \sum_j d(j, a^*)\\
			&\leq \textsc{Sum-Max}(A) + 2\sum_j d(j, a^*)\\
			&\leq \textsc{Sum-Max}(A) + 2\sum_i d(i, a_i)\\
			&= 3\cdot\textsc{Sum-Max}(A)
		\end{aligned}
		\]
		as desired. 
	\end{proof}
	\label{thm:3sssm}
\end{theorem}

\begin{theorem}
	The optimal solution of \textsc{Sum-Sum} is a 3 approximation for \textsc{Max-Max}.
	\begin{proof}
		Let $i = \operatorname{argmax}_{j \in \C} d(j, o_j)$. Similiar to the proof of Theorem \ref{thm:3sssm}, by triangle inequality, we have that 
		\[
		\begin{aligned}
			\textsc{Max-Max}(O) &= d(i, o_i)\\
			&\leq d(i, b_i) + d(b_i, o_i)\\
			&\leq \textsc{Max-Max}(B) + \frac{1}{n} \sum_j d(j, b_i) + \frac{1}{n} \sum_j d(j, o_i)\\
			&\leq \textsc{Max-Max}(B) + \frac{1}{n} \sum_j d(j, b^*) + \frac{1}{n} \sum_j d(j, o^*)\\
		\end{aligned}
		\]
		The last inequality is true due to how $b^*$ and $o^*$ are defined. Now, by Lemma \ref{lemma:sumsum}, we have $\sum_{i \in \C}d(i, o^*) \leq \sum_{i \in \C}d(i, b^*)$, therefore, 
		\[
		\begin{aligned}
			\textsc{Max-Max}(O) &\leq \textsc{Max-Max}(B) + \frac{1}{n} \sum_j d(j, b^*) +  \frac{1}{n} \sum_j d(j, b^*)\\
			&\leq \textsc{Max-Max}(B) + \frac{2}{n}\sum_j d(j, b^*)\\
			&\leq \textsc{Max-Max}(B) + 2 \max_{j \in \C}d(j, b_j)\\
			&= 3\cdot\textsc{Max-Max}(B)
		\end{aligned}
		\]
		as desired. 
	\end{proof}
	\label{thm:3ssmm}
\end{theorem}

Finally, combining Corollary \ref{coro:3ssms}, Theorem \ref{thm:3sssm}, Theorem \ref{thm:3ssmm}, we can see that 

\begin{theorem}
	The optimal solution of \textsc{Sum-Sum} is a 3 approximation for \textsc{Max-Sum}, \textsc{Max-Max} and \textsc{Sum-Max}.
	\label{thm:3ssall}
\end{theorem}

The above result gives a solution that is a 3 approximation for all the combinations of the 4 objectives defined in Definition \ref{def:objs}. In the next section, we will see if we can improve the simultaneous approximation ratio for some of the pairs of the 4 objectives.

\section{Improved Approximations for Pairs of Objectives}
\subsection{\textsc{Sum-Sum} and \textsc{Max-Sum}} \label{sec:ssms}

Note that so far when choosing multiple facilities, to get a good approximation ratio for two objectives, we have been using the optimal solution for one of the objectives. However, can we get a better approximation if we choose a subset of one of the optimal solutions and a subset of the other optimal solution? To see what it would result in, we first consider \textsc{Max-Sum} and \textsc{Sum-Sum}, and establish the following helpful notation.

\begin{definition}
Suppose $A,B$ are two non-empty sets in metric space $(M, d)$, define $f:M \times M \rightarrow \mathbb{R^{+}}\cup \{0\}$ as follows:
$$f(A,B) = \sum_{a\in A} \sum_{b\in B} d(a,b).$$
\end{definition}

\begin{lemma}
For any non-empty sets $A, B, C$ in metric space $(M, d)$, we have that $$f(A,B) \leq \frac{|B|}{|C|}f(A,C) + \frac{|A|}{|C|}f(B,C).$$
\begin{proof}
	Consider any $a \in A, b \in B, c \in C$, since $a,b,c$ are all in some metric space, we have that 
    $$f(a,b) \leq f(a,c) + f(c,b)$$
    Since the choice of $c$ is arbitrary, we also have that
    $$f(a,b) \leq \frac{1}{|C|}f(a,C) + \frac{1}{|C|}f(C, b).$$
    Then, we can see that
    \[
    \begin{aligned}
    	f(A,B) &= \sum_{a\in A}\left(\sum_{b \in B} f(a,b)\right)\\
    	&\leq \sum_{a\in A}\sum_{b \in B} \left(f(a,c) + f(c, b)\right)\\
    	&\leq \sum_{a\in A}\sum_{b \in B}\left(\frac{1}{|C|}f(a,C) + \frac{1}{|C|}f(C, b)\right)\\
    	&= \sum_{a\in A}\left(\frac{|B|}{|C|}f(a,C) + \frac{1}{|C|}f(C, B)\right)\\
    	&= \frac{|B|}{|C|}f(A,C) + \frac{|A|}{|C|}f(B,C)
    \end{aligned}
    \]
    as desired.
\end{proof}
\label{lemma:ratio}
\end{lemma}

Let $\optss$ be the optimal solution for $\css$, $\optms$ be the optimal solution for $\cms$, $O = \optms \cap \optss$, and $k' = k-|O|$, where $k$ is the number of facilities we are choosing. In order to form a solution which simultaneously approximates both objectives, we will consider 3 cases, (i) $k' > 1$ and $k'$ is even, (ii) $k' > 1$ and $k'$ is odd, (iii) $k' = 1$. First, we consider case (i), assume $k' > 1$ is even. 
For any $A\subseteq \F\setminus O$ and $|A|=k'$, we reorder each $a^i \in A$ such that $\sum_{a \in \C} d(a, a^1) \leq \sum_{a \in \C} d(a, a^2) \leq \cdots \leq \sum_{a \in \C} d(a, a^{k'})$. Then, let $Q_{\ms} = \{o_{\ms}^1, o_{\ms}^2, \cdots, o_{\ms}^{k'/2} \}, Q_{\sums} = \{o_{\sums}^1, o_{\sums}^2, \cdots, o_{\sums}^{k'/2} \}$. 
Note that the above also means that $Q_{\ms} = \operatorname{argmin}_{A\subseteq \optms\setminus O :|A| = k'/2}$$\sum_{i\in \C}f(i,A)$, 
and $Q_{\sums} = \operatorname{argmin}_{A\subseteq \optss\setminus O :|A| = k'/2}\sum_{i\in \C}f(i,A)$. We can then show the following result. 

\begin{theorem}
    For $k' > 1$ and even, there always exists a facility location set $A\subseteq\F$ such that choosing $A$ would give a $(5+\sqrt{17})/4$ approximation both for minimizing $\cms$ and minimizing $\css$. 

\begin{proof}
    First, note that by Lemma \ref{lemma:sumf} and Theorem \ref{theorem:2_gen} we have that 
    $$\ams(\optss) \leq \frac{1}{\ass(\optms)} + 2$$
    which implies that we can obtain a $\min(\ass(\optms), {1}/({\ass(\optms)} + 2))$ approximation by choosing $\optms$ or $\optss$. Now, in addition to $\optms$ and $\optss$, we will look at solution $A = O \cup Q_{\ms} \cup Q_{\sums}$. Let $R_{\ms} = \optms\setminus (Q_{\ms}\cup O), R_{\sums} = \optss\setminus (Q_{\sums}\cup O)$. Note that $A$ is composed of $\optms$ and $\optss$ with no duplicate facility locations, it is a valid solution. Then, we will analyze how good it is w.r.t. to $\cms(\optms)$. With triangle inequality, we have that
    \[
    \begin{aligned}
        \cms(A) &= \max_{i \in \C}f(i, A)\\
        &= \max_{i \in \C}f(i, O \cup Q_{\ms } \cup Q_{\sums})\\
        &= \max_{i \in \C} \left(  f(i, O) + f(i, Q_{\ms}) + f(i, Q_{\sums})\right) 
     \end{aligned}
     \]
     Here note that by Lemma \ref{lemma:ratio}, we have that $f(i, Q_{\ms }) \leq f(i, R_{\ms}) + (2/k')f(R_{\ms },Q_{\sums}),\ f(R_{\ms}, Q_{\sums}) \leq ({k'}/{2}) f(i, R_{\ms }) + ({k'}/{2}) f(i, Q_{\sums})$. Therefore, we can see that
     \[
     \begin{aligned}
        \cms(A) &\leq \max_{i \in \C} \left( f(i, O)+  f(i, Q_{\ms}) + f(i, R_{\ms}) +\frac{2}{k'}\cdot f(R_{\ms}, Q_{\sums})\right)\\
        &\leq \cms(\optms) +  \frac{2}{k'}\cdot f(R_{\ms}, Q_{\sums})\\
        &\leq  \cms(\optms) + \frac{2}{k'}\cdot \frac{1}{n} \left(\frac{k'}{2}\cdot \sum_{i \in \C} f(i, R_{\ms}) + \frac{k'}{2}\cdot\sum_{i \in \C} f(i, Q_{\sums})\right)\\
        &\leq  \cms(\optms) +  \frac{1}{n} \sum_{i \in \C} f(i, R_{\ms}) + \frac{1}{n} \sum_{i \in \C}f(i, Q_{\sums})
    \end{aligned}
    \]
    Note that since $R_{\ms} \subset \optms$, so we can see that 
    \[
    \begin{aligned}
        \sum_{i \in \C} f(i, R_{\ms}) &\leq \sum_{i \in \C} f(i, \optms)\\
        &\leq n\cdot \max_{i \in \C} f(i, \optms)\\
        &\leq n\cdot \cms(\optms)
    \end{aligned}
    \]
    Now, recall that we have $Q_{\sums} = \operatorname{argmin}_{A\subseteq \optss\setminus O :|A| = k'/2}\sum_{i\in \C}f(i,A)$, this means that $\sum_{i \in \C} f(i, Q_{\sums}) \leq \sum_{i \in \C} f(i, R_{\sums})$. We can then see that 
    \[
    \begin{aligned}
        \sum_{i \in \C}f(i, Q_{\sums}) &= \frac{1}{2} \left( \sum_{i \in \C} f(i, Q_{\sums}) + \sum_{i \in \C} f(i, Q_{\sums}) \right)\\
        &\leq \frac{1}{2} \left( \sum_{i \in \C} f(i, Q_{\sums}) + \sum_{i \in \C} d(i, R_{\sums}) \right)\\
        &= \frac{1}{2}\sum_{i \in \C}f(i, \optss\setminus O)\\
        &\leq \frac{1}{2} \sum_{i \in \C}f(i, \optss)\\
        &\leq \frac{1}{2} \css(\optss)
    \end{aligned}
    \]
    Now, since we have that $\sum_{i \in \C}f(i, Q_{\sums}) \leq \frac{1}{2} \css(\optss)$ and $\sum_{i \in \C} f(i, R_{\ms}) \leq {n} \cdot \cms(\optms)$, we can derive that 
\[
    \begin{aligned}
        \cms(A) 
        &\leq  \cms(\optms) + \frac{1}{n} \sum_{i \in \C} f(i, R_{\ms}) + \frac{1}{n} \sum_{i \in \C} f(i, Q_{\sums}) \\
        &\leq \cms(\optms) + \frac{1}{n} \cdot n \cdot \cms(\optms)) + \frac{1}{2n} \css(\optss)\\
        & = 2\cdot\cms(\optms) + \frac{1}{2n} \cdot \frac{1}{\ass(\optms)}\css(\optms)\\
        &\leq 2\cdot\cms(\optms) + \frac{1}{2n} \cdot \frac{1}{\ass(\optms)} \cdot \frac{n}{1}\cms(\optms)\\
        &= \left(2+  \frac{1}{2\ass(\optms)} \right)\cms(\optms)
    \end{aligned}
    \] 
Note that we are calculating the approximation ratio for both $\cms$ and $\css$ and it could be the case that one is worse than the other. Since the final simultaneous approximation ratio would be the worse of the two, next, we will examine how good $A$ is w.r.t. $\css(\optss)$.
\[
    \begin{aligned}
        \css(A) &= \sum_{i \in \C}f(i, A)\\
        &= \sum_{i \in \C}f(i, Q_{\ms}) + \sum_{i \in \C}f(i, Q_{\sums}) + \sum_{i \in \C}f(i, O)
    \end{aligned}
    \]
    Here note that we have previously showed that $\sum_{i \in \C}f(i, Q_{\sums}) \leq \frac{1}{2}\sum_{i \in \C}f(i, \optss\setminus O)$, using the same argument we also have that $\sum_{i \in \C}f(i, Q_{\ms}) \leq \frac{1}{2}\sum_{i \in \C}f(i, \optms\setminus O)$. Therefore, 
    \[
    \begin{aligned}
        \css(A) 
        &\leq \frac{1}{2} \sum_{i \in \C}f(i, \optms\setminus O) + \frac{1}{2} \sum_{i \in \C}f(i, \optss\setminus O) + \sum_{i \in \C}f(i, O)\\
        &= \frac{1}{2} \css(\optms) + \frac{1}{2} \css(\optss)\\
        & = \frac{1}{2} \ass(\optms) \cdot \css(\optss) + \frac{1}{2} \css(\optss)\\
        &= \frac{1 + \ass(\optms)}{2} \cdot \css(\optss)
    \end{aligned}
    \]
    Hence we can conclude that choosing $A$ is a $\max \left( 2+  \frac{1}{2\ass(\optms)}, \frac{1 + \ass(\optms)}{2}\right)$ approximation of $\cms$ and $\css$. Then, this means that one of $\optms,\optss,A$ would give us a 
    $$\min \left(\ass(\optms), \frac{1}{\ass(\optms)} + 2, \max \left( 2+  \frac{1}{2\ass(\optms)}, \frac{1 + \ass(\optms)}{2}\right)\right)$$
approximation of $\cms$ and $\css$. Finally, by Lemma \ref{lemma:half-half-bound}, we have that one of $\optms,\optss,A$ would give us at most a $(5+\sqrt{17})/4$ approximation for both objectives. 
\end{proof}
\label{thm:half-half-bound}
\end{theorem}

\begin{lemma}
For $x \geq 1$, $$\min \left(x, \frac{1}{x} + 2, \max \left( 2+  \frac{1}{2x}, \frac{1 + x}{2}\right)\right)$$ 
is bounded above by $(5+\sqrt{17})/4$.

\begin{proof}
	Let $f_1 = x, f_2 = \frac{1}{x} + 2, f_3 = 2+  \frac{1}{2x}, f_4 = \frac{1 + x}{2}.$ Suppose $x \geq 1$, we have that $f_1, f_4$ are monotone increasing and $f_2, f_3$ are monotone decreasing. In addition, note that we must have $f_2 = \frac{1}{x} + 2 > 2+  \frac{1}{2x} = f_3$ when $x \geq 1$, these two function would never intersect on $x \geq 1$. In addition, note that when $x \geq 3/2 + \sqrt{13}/2$, we have $f_3 < f_4$. Therefore, to find the maximum value of $\min(f_1, f_2, \max(f_3, f_4))$, we consider two cases, (i) $x \geq 3/2 + \sqrt{13}/2$, $\min(f_1, f_2, \max(f_3, f_4)) = \min(f_1, f_2, f_4)$, (ii) $x < 3/2 + \sqrt{13}/2$, $\min(f_1, f_2, \max(f_3, f_4)) = \min(f_1, f_2, f_3)$.
	
	First, consider case (i). Here note that when $f_1 = f_4$, we have $x = 1$, therefore $f_4 < f_1$ when $x \geq 1$. Hence this case reduce to $\min(f_2, f_4)$. We claim that one of $f_2$ and $f_4$ must be less than or equal to $(5+\sqrt{17})/4$. To see this, suppose $f_2 \geq (5+\sqrt{17})/4$, when $x \geq 3/2 + \sqrt{13}/2$, this gives us $3/2 + \sqrt{13}/2 \leq x \leq 3/2 + \sqrt{17}/2$, then $f_4 = (1+x)/2 \leq 5/4 + \sqrt{17}/4 = (5+\sqrt{17})/4$. Similarly, suppose $f_4 \geq (5+\sqrt{17})/4$, when $x > 3/2 + \sqrt{13}/2$, this gives us $x \geq 3/2 + \sqrt{17}/2$, then $f_2 = 1/x + 2 \leq 1/(3/2 + \sqrt{17}/2) + 2 = (5+\sqrt{17})/4$. Therefore we can conclude that when $x \geq 3/2 + \sqrt{13}/2$, $\min(f_1, f_2, \max(f_3, f_4)) = \min(f_1, f_2, f_4) \leq (5+\sqrt{17})/4$.
	
	Then, consider case (ii). Recall that we have shown that $f_2 > f_3$ for $x \geq 1$, this case reduce to $ \min(f_1, f_3)$. We claim that one of $f_1$ and $f_3$ must be less than or equal to $1 + \sqrt{3/2}$. To see this, suppose $f_1 = x \geq 1 + \sqrt{3/2}$, when $x \geq 3/2 + \sqrt{13}/2$, otherwise we would have $f_1 = x < 1 + \sqrt{3/2}$. This means that $f_3 = 2 + 1/(2x) \leq 2 + 1/(2 + \sqrt{6}) = 1 + \sqrt{3/2}$. Therefore we can conclude that when $x < 3/2 + \sqrt{13}/2$, $\min(f_1, f_2, \max(f_3, f_4)) = \min(f_1, f_2, f_3) \leq 1 + \sqrt{3/2}$.
	
	Now, combine the above two cases together, we have that 
	$$\min(f_1, f_2, \max(f_3, f_4)) \leq \max (1 + \sqrt{\frac{3}{2}}, \frac{5+\sqrt{17}}{4}) = \frac{5+\sqrt{17}}{4}$$
	as desired. 
\end{proof}
\label{lemma:half-half-bound}
\end{lemma}

Now, we consider the case (ii), where $k' > 1$ is odd. Similar to case (i), let $Q_{\ms} = \{o_{\ms}^1, o_{\ms}^2, \cdots,$\\$o_{\ms}^{(k'-1)/2} \}$, $Q_{\sums} = \{o_{\sums}^1, o_{\sums}^2, \cdots, o_{\sums}^{(k'+1)/2} \}$. This means that $$Q_{\ms} = \operatorname{argmin}_{A\subseteq \optms\setminus O :|A| = (k'-1)/2}\sum_{i\in \C}f(i,A),$$ $$Q_{\sums} = \operatorname{argmin}_{A\subseteq \optss\setminus O :|A| = (k'+1)/2}\sum_{i\in \C}f(i,A).$$

\begin{theorem}
    For $k' > 1$ and odd, there always exists a facility location set $A\subseteq\F$ such that choosing $A$ would give a $1+\sqrt{5/3}$ approximation both for minimizing $\cms$ and minimizing $\css$. 

\begin{proof}
    Similar to Theorem \ref{thm:half-half-bound}, first, note that by Lemma \ref{lemma:sumf} and Theorem \ref{theorem:2_gen} we have that 
    $$\ams(\optss) \leq \frac{1}{\ass(\optms)} + 2$$
    which implies that we can obtain a $\min(\ass(\optms), {1}/({\ass(\optms)} + 2))$ approximation by choosing $\optms$ or $\optss$. Now, in addition to $\optms$ and $\optss$, we will look at solution $A = O \cup Q_{\ms} \cup Q_{\sums}$. Let $R_{\ms} = \optms\setminus (Q_{\ms}\cup O), R_{\sums} = \optss\setminus (Q_{\sums}\cup O)$. Note that $A$ is composed of $\optms$ and $\optss$ with no duplicate facility locations, it is a valid solution. Then, we will analyze how good it is w.r.t. to $\cms(\optms)$. Here note that by Lemma \ref{lemma:ratio} and $k' \geq 3$, we have that $f(i, Q_{\sums}) \leq f(i, R_{\ms}) + \frac{2}{k'+1}\cdot f(R_{\ms}, Q_{\sums}),\ f(R_{\ms}, Q_{\sums}) \leq \frac{k'+1}{2}\cdot f(i, R_{\ms}) + \frac{k'+1}{2}\cdot f(i, Q_{\sums})$ With triangle inequality, we have that
    \[
    \begin{aligned}
        \cms(A) &= \max_{i \in \C}f(i, A)\\
        &= \max_{i \in \C}f(i, O \cup Q_{\ms} \cup Q_{\sums})\\
        &= \max_{i \in \C} \left(  f(i, O) + f(i, Q_{\ms}) + f(i, Q_{\sums})\right) \\
        &\leq \max_{i \in \C} \left( f(i, O)+  f(i, Q_{\ms}) + f(i, R_{\ms}) + \frac{2}{k'+1}\cdot f(R_{\ms}, Q_{\sums})\right)\\
        &\leq \cms(\optms) +  \frac{2}{k'+1}\cdot f(R_{\ms}, Q_{\sums})\\
        &\leq  \cms(\optms) + \frac{1}{n}\cdot \frac{2}{k'+1}\left(\frac{k'+1}{2}\cdot\sum_{i \in \C} f(i, R_{\ms}) + \frac{k'+1}{2}\cdot\sum_{i \in \C} f(i, Q_{\sums})\right)\\
        &\leq  \cms(\optms) +  \frac{1}{n} \sum_{i \in \C} f(i, R_{\ms}) + \frac{1}{n} \sum_{i \in \C}f(i, Q_{\sums})
    \end{aligned}
    \]
    Recall that we have shown in Theorem \ref{thm:half-half-bound} that $\sum_{i \in \C} f(i, R_{\ms}) \leq n\cdot \cms(\optms)$.
    In addition, we have that 
    \[
    \begin{aligned}
        \sum_{i \in \C}f(i, Q_{\sums}) &= \frac{k'+1}{2k'} \sum_{i \in \C} f(i, Q_{\sums}) + \frac{k'-1}{2k'} \sum_{i \in \C} f(i, Q_{\sums}) 
    \end{aligned}
    \]
    Now, recall that we have $Q_{\sums} = \operatorname{argmin}_{A\subseteq \optss\setminus O :|A| = (k'-1)/2}\sum_{i\in \C}f(i,A)$, this means that $$\sum_{i \in \C} f(i, Q_{\sums}) \leq \frac{k'+1}{k'-1} \sum_{i \in \C} f(i, R_{\sums}).$$ We can then see that 
    \[
    \begin{aligned}
        \sum_{i \in \C}f(i, Q_{\sums}) 
        &\leq \frac{k'+1}{2k'} \sum_{i \in \C} f(i, Q_{\sums}) + \frac{k'-1}{2k'} \cdot \frac{k'+1}{k'-1} \sum_{i \in \C} f(i, R_{\sums})\\
        &= \frac{k'+1}{2k'}\sum_{i \in \C}f(i, \optss\setminus O)\\
        &\leq \frac{k'+1}{2k'} \sum_{i \in \C}f(i, \optss)\\
        &\leq \frac{k'+1}{2k'} \css(\optss)
    \end{aligned}
    \]
    Now, since we have that $\sum_{i \in \C}f(i, Q_{\sums}) \leq \frac{k'+1}{2k'} \css(\optss)$ and $\sum_{i \in \C} f(i, R_{\ms}) \leq {n} \cdot \cms(\optms)$, we can derive that 
\[
    \begin{aligned}
        \cms(A) 
        &\leq  \cms(\optms) + \frac{1}{n} \sum_{i \in \C} f(i, R_{\ms}) + \frac{1}{n} \sum_{i \in \C} f(i, Q_{\sums}) \\
        &\leq \cms(\optms) + \frac{1}{n} \cdot n \cdot \cms(\optms) + \frac{k'+1}{2k'n} \css(\optss)\\
        & = {2}\cdot \cms(\optms) + \frac{k'+1}{2k'n} \cdot \frac{1}{\ass(\optms)}\css(\optms)\\
        &\leq {2}\cdot \cms(\optms) + \frac{k'+1}{2k'n} \cdot \frac{1}{\ass(\optms)} \cdot \frac{n}{1}\cms(\optms)\\
        &= \left(2+  \frac{k'+1}{2k'}\cdot\frac{1}{\ass(\optms)} \right)\cms(\optms)
    \end{aligned}
    \] 
Now, since $k' \geq 3$, $\frac{k'+1}{2k'} \leq \frac{2}{3}$, we can observe that 
\[
\begin{aligned}
	\cms(A) &\leq \left(2+  \frac{k'+1}{2k'}\cdot\frac{1}{\ass(\optms))} \right)\cms(\optms)\\
    &\leq \left(2+  \frac{2}{3\ass(\optms)} \right)\cms(\optms)
\end{aligned}
\]

Note that we are calculating the approximation ratio for both $\cms$ and $\css$ and it could be the case that one is worse than the other. Since the final simultaneous approximation ratio would be the worse of the two, next, we will examine how good $A$ is w.r.t. $\css(\optss)$.
\[
    \begin{aligned}
        \css(A) &= \sum_{i \in \C}f(i, A)\\
        &= \sum_{i \in \C}f(i, Q_{\ms}) + \sum_{i \in \C}f(i, Q_{\sums}) + \sum_{i \in \C}f(i, O)
    \end{aligned}
    \]
    Here note that we have previously showed that $\sum_{i \in \C}f(i, Q_{\sums}) \leq \frac{k'+1}{2k'}\sum_{i \in \C}f(i, \optss\setminus O)$, using the same argument we also have that $\sum_{i \in \C}f(i, Q_{\ms}) \leq \frac{k'-1}{2k'}\sum_{i \in \C}f(i, \optms\setminus O)$. Therefore, 
    \[
    \begin{aligned}
        \css(A) 
        &\leq \frac{k'-1}{2k'} \sum_{i \in \C}f(i, \optms\setminus O) + \frac{k'+1}{2k'} \sum_{i \in \C}f(i, \optss\setminus O) + \sum_{i \in \C}f(i, O)\\
        &= \frac{k'-1}{2k'} \css(\optms) + \frac{k'+1}{2k'} \css(\optss)\\
        & = \frac{k'-1}{2k'} \ass(\optms) \cdot \css(\optss) + \frac{k'+1}{2k'} \css(\optss)\\
        &= \frac{k'+1 + (k'-1)\ass(\optms)}{2k'} \cdot \css(\optss)
    \end{aligned}
    \]
    Recall that we have $\ass(\optms) \geq 1$, so we have that 
    \[
    \begin{aligned}
        \css(A) 
        &\leq \frac{k'+ k'\ass(\optms) + (1 - \ass(\optms) )}{2k'} \cdot \css(\optss)\\
        &\leq \frac{k'+ k'\ass(\optms)}{2k'} \cdot \css(\optss)\\
        &= \frac{1+ \ass(\optms) }{2} \cdot \css(\optss)
    \end{aligned}
    \]

    Hence we can conclude that choosing $A$ is a $\max\left( 2+  \frac{2}{3\ass(\optms)}, \frac{1 + \ass(\optms)}{2}\right)$ approximation of $\cms$ and $\css$. Then, this means that one of $\optms,\optss,A$ would give us a 
    $$\min\left(\ass(\optms), \frac{1}{\ass(\optms)} + 2, \max\left( 2+  \frac{2}{3\ass(\optms)}, \frac{1 + \ass(\optms)}{2}\right) \right)$$
approximation of $\cms$ and $\css$. Finally, similar to the proof for Lemma \ref{lemma:half-half-bound}, we have that one of $\optms,\optss,A$ would give us at most a $1+ \sqrt{5/3}$ approximation for both objectives.     
\end{proof}
\label{thm:half-half-bound-odd}
\end{theorem}

Therefore, with Theorem \ref{thm:half-half-bound} and Theorem \ref{thm:half-half-bound-odd}, we can conclude that the upper bound of approximation ratio for both $\cms$ and $\css$ when $k' > 1$ is $(5+\sqrt{17})/4$ when $k'$ is even and it is $1+\sqrt{5/3}$ when $k'$ is odd. 

Finally, we consider the remaining case, when $k' = 1$. In this case $\optss$ and $\optms$ overlap except for one location. We can show that for $k > 2$ the upper bound is much better than $1 + \sqrt{2}$ for approximating both $\optms$ and $\optss$. However, the  $1 + \sqrt{2}$ bound is still tight when $k = 2$. 

\begin{theorem}
No deterministic algorithm can approximate both $\cms$ and $\css$ simultaneously within a factor of better than $1+\sqrt{2}$ when $k=2$.
\label{thm:2_bound}
\end{theorem}

\begin{proof}
    Consider the case where there are four locations on a line as shown in Figure \ref{fig:ex2}. There is one client on location $A$, $n-1$ clients on location $B$, one facility location at $A$, one facility location at $B$ and one facility location at $C$ (i.e. the multiset of facility locations is $\F = \{A, B, C\}$). We also have that $d(A,B) = 1$, $d(C,B) =\sqrt{2}-1$. Note that in this case $\optms$ is taking facilities on $A,B$, $\optss$ is taking facilities on $B,C$, and we also consider $H = \{A,C\}$,  which consists one facility in $\optss$ and one in $\optms$. We obtain the values as shown in Table \ref{tab:ex2}.

    \begin{figure}[t]
  \begin{center}
        \begin{tikzpicture}[
            roundnode/.style={rectangle, draw=blue!60, fill=green!5, very thick, minimum size=5mm},
            squarednode/.style={rectangle, draw=blue!60, fill=blue!5, very thick, minimum size=5mm},
            scale = 1, transform shape
        ]
        
            \node[roundnode, label=above:{$1$ client, 1 facility}] (A) {$A$};
            \node[roundnode, label=above:{$n-1$ clients, $1$ facility}] (B) [right = 5cm of A]{$B$};
            \node[squarednode, label=above:{1 facility}] (C) [right = 3cm of B]{$C$};

            \draw (A) -- (B) node[draw=none,fill=none,font=\scriptsize,midway,below ] {$1$};

            \draw (B) -- (C)
            node[draw=none,fill=none,font=\scriptsize,midway,below ] {$\sqrt{2}-1$};

        \end{tikzpicture}
  
  \caption{An instance with two client locations A, B and three possible facility locations A, B, C with distances between any two adjacent locations as well as the number of clients and/or facilities at each location labeled.}
  \label{fig:ex2}
  \end{center}
\end{figure}
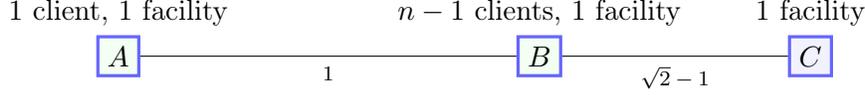

\begin{table}[]
\centering
\begin{tabular}{|c|c|c|c|}
\hline
       & $\optss$                              & $\optms$ & $H$                \\ \hline
$\css$ & $(\sqrt{2}-1)n+1$                     & n        & $\sqrt{2} \cdot n$ \\ \hline
$\cms$ & $1 + \sqrt{2}$                        & 1        & $\sqrt{2}$         \\ \hline
$\ass, n\rightarrow \infty$ & 1 & {$1 + \sqrt{2}$} & {$2 + \sqrt{2}$} \\ \hline
$\ams, n\rightarrow \infty$ & {$1 + \sqrt{2}$} & 1        & $\sqrt{2}$         \\ \hline
\end{tabular}
\caption{An instance for Theorem \ref{thm:2_bound},with $\css$, $\cms$, $\ass$, $\ams$ listed for $\optss$, $\optms$ and $H$.}
\label{tab:ex2}
\end{table}
    
    By choosing either $\optss$ or $\optms$ gives us a $1+\sqrt{2}$ for approximating both objectives and choosing $H$ gives a worse approximation ratio. Therefore, there does not exist a solution that would give an approximation ratio better than $1+\sqrt{2}$.     
\end{proof}

Next, to show that we can obtain a better upper bound for $k > 2$, we will first show a more general result. 
\begin{theorem}
    For $k' \geq 1$, given the optimal facility set $\optms$ that minimizes $\cms$ and the optimal facility set $\optss$ that minimizes $\css$, we have that $\ams(\optss) \leq 1 + \frac{k'}{k-k'}+\max\left(1, \frac{k'}{k-k'}\right)\cdot\frac{1}{\ass\left(\optms\right)}$. 

\begin{proof}
    First, let $A = \optms\setminus O, C = \optss\setminus O$. Note that we also have $|A|=|C|=k', |O| = k-k'$. With triangle inequality, we have that
    \[
    \begin{aligned}
        \cms(\optss) &= \max_{i \in \C}f(i, O \cup C)\\
        &= \max_{i \in \C} \left(  f(i, O) + f(i, C)\right) \\    
    \end{aligned}
    \]
 Now, by Lemma \ref{lemma:ratio}, we have that $f(i, C) \leq \frac{k'}{k-k'}f(i, O) + \frac{1}{k-k'}f(O, C)$, $f(O, C) \leq k' \cdot f(i, O) + (k-k') f(i, C)$ for any $i \in \C$, so we can see that  
    \[
    \begin{aligned}
        \cms(\optss) &\leq \max_{i \in \C} \left( f(i, O)+  \frac{k'}{k-k'}f(i, O) + \frac{1}{k-k'}f(O, C)\right)\\
        &\leq \left(1 + \frac{k'}{k-k'}\right)\cms(\optms) +  \frac{1}{k-k'}f(O, C)\\
        &\leq  \left(1 + \frac{k'}{k-k'}\right)\cms(\optms) + \frac{1}{k-k'} \cdot \frac{1}{n} \left(k'\cdot\sum_{i \in \C} f(i, O) + (k-k')\sum_{i \in \C} f(i, C)\right)\\
        &\leq  \left(1 + \frac{k'}{k-k'}\right)\cms(\optms) + \max\left(1, \frac{k'}{k-k'}\right)\frac{1}{n} \left(\sum_{i \in \C} f(i, O) + \sum_{i \in \C} f(i, C)\right)\\
        &\leq  \left(1 + \frac{k'}{k-k'}\right)\cms(\optms) + \max\left(1, \frac{k'}{k-k'}\right)\frac{1}{n} \cdot \css(\optss)\\
        &\leq  \left(1 + \frac{k'}{k-k'}\right)\cms(\optms) + \max\left(1, \frac{k'}{k-k'}\right)\frac{1}{n} \cdot\frac{\css(\optms)}{\ass(\optms)}\\
        &\leq  \left(1 + \frac{k'}{k-k'}\right)\cms(\optms) + \max\left(1, \frac{k'}{k-k'}\right)\frac{n}{n} \cdot\frac{\cms(\optms)}{\ass(\optms)}
    \end{aligned}
    \]
    Then, we divide both sides by $\cms(\optms)$, which gives us 
    \[
    \begin{aligned}
    	\ams(\optss) &\leq  1 + \frac{k'}{k-k'} + \max\left(1, \frac{k'}{k-k'}\right)\cdot\frac{1}{\ass(\optms)}
    \end{aligned}
    \]
    
\end{proof}
\label{thm:overlap-bound}
\end{theorem}

\begin{corollary}
Let $a = \frac{k'}{k-k'}$. For $a \leq 1$, there always exists a facility set $A\subseteq\F$ such that choosing $A$ would give a $(\sqrt{a^2 + 2a + 5}+a+1)/2$ approximation both for minimizing $\cms$ and minimizing $\css$. In fact, we would either get a $(1, (\sqrt{a^2 + 2a + 5}+a+1)/2)$ approximation by choosing $\optms$ or a $((\sqrt{a^2 + 2a + 5}+a+1)/2, 1)$ approximation by choosing $\optss$. In other words, at least one of $\ass(\optms)$ or $\ams(\optss)$ is always less than or equal to $(\sqrt{a^2 + 2a + 5}+a+1)/2$. 

\begin{proof}

First note that since $a \leq 1$, by Theorem \ref{thm:overlap-bound}, we have that $\ams(\optss) \leq  1 + a + \frac{1}{\ass(\optms)}$.
We will show that $\max (\min (\ass(\optms), \ams(\optss))) \leq (\sqrt{a^2 + 2a + 5}+a+1)/2$. 
To do this, we will consider two cases. First, assume $\ass(\optms) > (\sqrt{a^2 + 2a + 5}+a+1)/2$, then we would have $\ams(\optss) \leq 1 + a + \frac{1}{\ass(\optms)} < 1 + a + \frac{2}{\sqrt{a^2 + 2a + 5}+a+1} = (\sqrt{a^2 + 2a + 5}+a+1)/2$. Otherwise, we would have $\ass(\optms) \leq (\sqrt{a^2 + 2a + 5}+a+1)/2$. Therefore, we can conclude that $\max (\min (\ass(\optms), \ams(\optss))) \leq (\sqrt{a^2 + 2a + 5}+a+1)/2$. This result implies that one of $\ams(\optss)$ and $\ass(\optms)$ is less than or equal to $(\sqrt{a^2 + 2a + 5}+a+1)/2$, as desired.
\end{proof}
\label{coro:2_gen}
\end{corollary}

\begin{corollary}
For $k' = 1, k \geq 3$, there always exists a facility location set $A\subseteq\F$ such that choosing $A$ would give a $2$ approximation both for minimizing $\cms$ and minimizing $\css$. 
\label{coro:3_less}
\end{corollary}

\begin{proof}
Let $a = \frac{1}{k-1}$, note that since $k\geq 3$, we have $a \leq \frac{1}{2}< 1 $. Then, by Corollary \ref{coro:2_gen}, we have that there always exists a facility location $A\in\F$ such that choosing $A$ would give a $(\sqrt{a^2 + 2a + 5}+a+1)/2$ approximation both for minimizing $\cms$ and minimizing $\css$. Then, we have that 
\[
\begin{aligned}
	\frac{\sqrt{a^2 + 2a + 5}+a+1}{2} \leq \frac{\sqrt{(1/2)^2 + 2\cdot (1/2) + 5}+(1/2)+1}{2} = 2
\end{aligned}
\]
as desired. 
\end{proof}

This means that if we are choosing more than 2 facility locations, if the optimal solution for $\cms$ and $\css$ are the same except for one location, then either $\optms$ or $\optss$ is a 2 approximation for both objectives. In addition, since 2 is smaller than the approximation ratio we found in Theorem \ref{thm:half-half-bound} and Theorem \ref{thm:half-half-bound-odd}, these two theorems hold for any $k'$ when $k > 2$ (note that $k' = 0$ implies that $\optms = \optss$, therefore it's a 1 approximation for both objectives by choosing $\optss$).

Recall that by Corollary \ref{coro:bound_2}, using previous results and choosing either $\optms$ or $\optss$, we cannot get an approximation ratio better than $1 + \sqrt{2}\approx 2.41$. By choosing the best of these two solutions together with $A = O \cup Q_{\ms} \cup Q_{\sums}$, however, we are able to use Theorems \ref{thm:half-half-bound}, \ref{thm:half-half-bound-odd}, and \ref{coro:3_less} to improve this bound to $\max\{(5+\sqrt{17})/4, 1 + \sqrt{5/3}, 2\} = 1 + \sqrt{5/3} \approx 2.29$ for any $k\geq 3$. In addition, we can show that when we are choosing more than 2 facilities, the lower bound is $(4+\sqrt{7})/3 \approx 2.215$.

\begin{theorem}
No deterministic algorithm can approximate both $\cms$ and $\css$ simultaneously within a factor of better than $(4+\sqrt{7})/3$ for $k \geq 3$.
\label{thm:sum_low_bound}
\end{theorem}
\begin{proof}

    Consider the case where there are four locations on a line as shown in Figure \ref{fig:ex1}. 
    There is one client on location $A$, $n-1$ clients on location $C$, three facility locations at $B$ and three facility locations at $D$ (i.e. the multiset of facility locations is $\F = \{B,B,B,D,D,D\}$). We also have that $d(A,B) = d(B,C) = 1$, $d(C,D) = (\sqrt{7}-2)/3$. Note that in this case $\optms$ is taking all facilities on $B$, $\optss$ is taking all facilities on $D$, and we also consider $H$ which consists one facility in $\optss$ and two in $\optms$ and $H'$ which consists two facility in $\optss$ and one in $\optms$. We obtain the values as shown in Table \ref{tab:ex5}.

    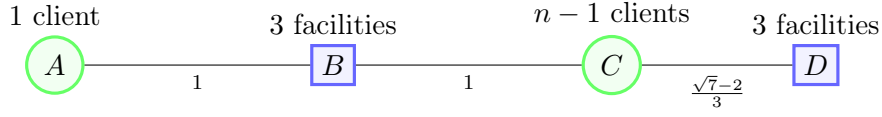
\begin{figure}[t]
  \begin{center}
  \begin{tikzpicture}[
            roundnode/.style={circle, draw=green!60, fill=green!5, very thick, minimum size=7mm},
            squarednode/.style={rectangle, draw=blue!60, fill=blue!5, very thick, minimum size=5mm},
            scale = 1, transform shape
        ]
        
            \node[roundnode, label=above:{$1$ client}] (A) {$A$};
            \node[squarednode, label=above:{$3$ facilities}] (B) [right = 3cm of A]{$B$};
            \node[roundnode, label=above:{$n-1$ clients}] (C) [right = 3cm of B]{$C$};
            \node[squarednode, label=above:{$3$ facilities}] (D) [right = 2cm of C]{$D$};
            
            \draw (A) -- (B) node[draw=none,fill=none,font=\scriptsize,midway,below ] {$1$};

            \draw (B) -- (C)
            node[draw=none,fill=none,font=\scriptsize,midway,below ] {$1$};
            
            \draw (C) -- (D)
            node[draw=none,fill=none,font=\scriptsize,midway,below ] {$\frac{\sqrt{7}-2}{3}$};
            
        \end{tikzpicture}
  \caption{An instance with two client locations A, C and two possible facility locations B, D such that $B$ and $D$ can both be chosen 3 times individually, with distances between any two adjacent locations as well as the number of clients or facilities at each location labeled.}
  \label{fig:ex1}
  \end{center}
\end{figure}

\begin{table}[t]
\centering
\begin{tabular}{|c|c|c|c|c|}
\hline
       & $\optss$     & $\optms$ & $H$                        & $H'$         \\ \hline
$\css$ &
  $(\sqrt{7}-2)n + 6$ &
  $3n$ &
  $\left(\frac{2\sqrt{7}-1}{3}\right)n + 6$ &
  $\left(\frac{\sqrt{7}+4}{3}\right)n + 2$ \\ \hline
$\cms$ & $4+\sqrt{7}$ & $3$      & $\frac{11 + 2\sqrt{7}}{3}$ & $2+\sqrt{7}$ \\ \hline
$\ass, n\rightarrow \infty$ &
  1 &
  {$2+\sqrt{7}$} &
  {$\frac{4+\sqrt{7}}{3}$} &
  $\frac{5+2\sqrt{7}}{3}$ \\ \hline
$\ams, n\rightarrow \infty$ &
  {$\frac{4+\sqrt{7}}{3}$} &
  1 &
  {$\frac{11 + 2\sqrt{7}}{9}$} &
  $\frac{2+\sqrt{7}}{3}$ \\ \hline
\end{tabular}
\caption{An instance for Theorem \ref{thm:sum_low_bound}, with $\css$, $\cms$, $\ass$, $\ams$ listed for $\optss$, $\optms$, $H$ and $H'$.}
\label{tab:ex5}
\end{table}
    
    Note that in this case, the only options of locations are both facilities on $B$ ($\optms$), both facilities on $D$ ($\optss$) and a mixture of facilities on $B$ and $D$ ($H$ and $H'$). By choosing either $\optss$ or $H$ gives us a $(4+\sqrt{7})/3$ for approximating both objectives and choosing $\optms$ or $H'$ gives a worse approximation ratio. Therefore, there does not exist a solution that would give an approximation ratio better than $(4+\sqrt{7})/3$ when $k \geq 3$. 
\end{proof}

\subsection{\textsc{Sum-Max} and \textsc{Max-Max}}

By Lemma \ref{lemma:maxf} and Corollary \ref{coro:bound_2}, there exists a solution that is a $1+\sqrt{2}$ approximation for both \textsc{Sum-Max} and \textsc{Max-Max}. 
Unfortunately, unlike what we showed for $\cms$ and $\css$ in the previous section, the best possible upper bound for
simultaneously approximating $\cmm$ and $\csm$ is still $1 + \sqrt{2}$, even for larger values of $k$. 

\begin{theorem}
    For any deterministic algorithm, there does not exist an approximation bound better than $1 + \sqrt{2}$ that simultaneously approximates $\cmm$ and $\csm$.
\begin{proof}
    Consider the case where there are four locations on a line as shown in Figure \ref{fig:ex3}. 
    There is one client on location $A$, $n-1$ clients on location $C$, $k$ facility locations at $B$ and $k$ facility locations at $D$ (i.e the multiset of facility locations is $\F = \{\underbrace{B, \ldots, B}_k , \underbrace{D, \ldots, D}_k\}$). We also have that $d(A,B) = d(B,C) = 1$, $d(C,D) =  \sqrt{2}-1$. Note that in this case $\optmm$ is taking both facilities on $B$, $\optsm$ is taking both facilities on $D$, and we also consider $H$ which consists at least one facility in $O_n$ and at least one facility in $O_1$. We obtain the values as shown in Table \ref{tab:ex3}.

\begin{figure}[t]
  \begin{center}
  \begin{tikzpicture}[
            roundnode/.style={circle, draw=green!60, fill=green!5, very thick, minimum size=7mm},
            squarednode/.style={rectangle, draw=blue!60, fill=blue!5, very thick, minimum size=5mm},
            scale = 1, transform shape
        ]
        
            \node[roundnode, label=above:{$1$ client}] (A) {$A$};
            \node[squarednode, label=above:{$k$ facilities}] (B) [right = 3cm of A]{$B$};
            \node[roundnode, label=above:{$n-1$ clients}] (C) [right = 3cm of B]{$C$};
            \node[squarednode, label=above:{$k$ facilities}] (D) [right = 2cm of C]{$D$};
            
            \draw (A) -- (B) node[draw=none,fill=none,font=\scriptsize,midway,below ] {$1$};

            \draw (B) -- (C)
            node[draw=none,fill=none,font=\scriptsize,midway,below ] {$1$};
            
            \draw (C) -- (D)
            node[draw=none,fill=none,font=\scriptsize,midway,below ] {$\sqrt{2}-1$};
            
        \end{tikzpicture}
  \caption{An instance with two client locations A, C and two possible facility locations B, D such that both $B$ and $D$ can be chosen $k$ times individually, with distances between any two adjacent locations as well as the number of clients or facilities at each location labeled.}
  \label{fig:ex3}
  \end{center}
\end{figure}
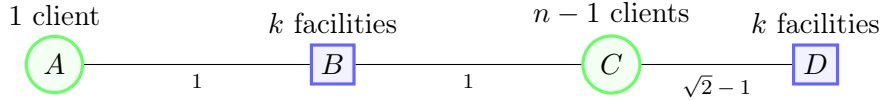

\begin{table}[]
\centering
\begin{tabular}{|c|c|c|c|}
\hline
       & $\optsm$                              & $\optmm$                              & $H$                                   \\ \hline
$\csm$ & $(\sqrt{2}-1)n+2$                     & n                                     & $*$                        \\ \hline
$\cmm$ & $1 + \sqrt{2}$                        & 1                                     & $1 + \sqrt{2}$                        \\ \hline
$\asm, n\rightarrow \infty$ & 1                                     & {$1 + \sqrt{2}$} & {*} \\ \hline
$\amm, n\rightarrow \infty$ & {$1 + \sqrt{2}$} & 1                                     & {$1 + \sqrt{2}$} \\ \hline
\end{tabular}
\caption{An instance for Theorem \ref{thm:max_bound},with $\cmm$, $\csm$, 
 $\amm$, $\asm$ listed for $\optmm$, $\optsm$ and $H$. Note that * means that the value varies by the choice of $H$.}
\label{tab:ex3}
\end{table}

    Note that in this case, the only options of locations are both facilities on $B$ ($\optmm$), both facilities on $D$ ($\optsm$) and a combination of at least one facility on $B$ and at least one facility on $D$ ($H$). By choosing any of $\optsm$, $\optmm$ or $H$ gives us at least a $1 + \sqrt{2}$ approximation for approximating both objectives. Therefore, there does not exist a solution that would give an approximation ratio better than $1 + \sqrt{2}$. 
    
\end{proof}
\label{thm:max_bound}
\end{theorem}

\subsection{\textsc{Max-Max} and \textsc{Max-Sum}} \label{section:mmms}
Until now, we have only looked at pairs of objectives that have the same individual cost function for each client but different ones for the overall cost (max and sum respectively). We now consider pairs of objectives that have different individual cost functions for each client but the same cost function for the overall cost. Note that in this case, when $k = 1$, the two objectives would become exactly the same. For example, consider $\cmm$ and $\cms$: when $k=1$, both would become minimax, and thus would have the same optimum solution. Because of these differences, previous results from Section \ref{section:prelim} do not apply to this set of objectives, and different techniques are required to form a good approximation for both objectives. We will first consider \textsc{Max-Max} and \textsc{Max-Sum}.
Let $\optms$ be the optimal solution for $\cms$, and $\optmm$ be the optimal solution for $\cmm$.
\begin{definition} We define $k_1, k_2$ as follows:
	\begin{enumerate}
		\item $k_1 \cdot \max_{i \in \C} \max_{a \in \optms} d(i,a) = \max_{i \in \C} \sum_{a \in \optms} d(i,a)$
		\item $k_2 \cdot \max_{i \in \C} \max_{a \in \optmm} d(i,a) = \max_{i \in \C} \sum_{a \in \optmm} d(i,a)$
	\end{enumerate}
	Note that we have $1 \leq k_1 \leq k, 1 \leq k_2 \leq k$.
	\label{def_ks}
\end{definition}

First, we make the following simple observation.

\begin{theorem}
	For any $k$, given the optimal facility set $\optmm$ that minimizes $\cmm$ and the optimal facility set $\optms$ that minimizes $\cms$, we have that $\ams(\optmm) \cdot \amm(\optms) = \frac{k_2}{k_1}$. 
	
	\begin{proof}
		By Definition \ref{def_ks}, we have that 
		\[
		\begin{aligned}
			\max_{i \in \C} \max_{a \in \optms} d(i,a) &= \frac{1}{k_1} \max_{i \in \C} \sum_{a \in \optms} d(i,a)\\
			&= \frac{1}{k_1 \cdot \ams(\optmm)} \max_{i \in \C} \sum_{a \in \optmm} d(i,a)\\
			&= \frac{k_2}{k_1 \cdot \ams(\optmm)} \max_{i \in \C} \max_{a \in \optmm} d(i,a)
		\end{aligned}
		\]
	Then, divide both sides by $\max_{i \in \C} \max_{a \in \optmm} d(i,a)$, we have that 
	\[
	\begin{aligned}
		\amm(\optms) &= \frac{k_2}{k_1 \cdot \ams(\optmm)}\\
		\amm(\optms) \cdot \ams(\optmm) &= \frac{k_2}{k_1}
	\end{aligned}
	\]
	as desired.
	\end{proof}
	\label{thm:gen_mmms}
\end{theorem}

\begin{corollary}
	For any $k$, there always exists a facility location set $A \subseteq \F$ such that choosing $A$ would give a $\sqrt{k_2/k_1}$ approximation both for minimizing $\cmm$ and minimizing $\cms$. In fact, we would either get a $(1, \sqrt{k_2/k_1})$ approximation by choosing $\optmm$ or a $(\sqrt{k_2/k_1},1)$ approximation by choosing $\optms$. In other words, at least one of $\ams(\optmm)$ or $\amm(\optms)$ is always less than or equal to $\sqrt{k_2/k_1}$.
	\begin{proof}
		To show the above result, we will consider two cases. First, assume $\ams(\optmm) > \sqrt{k_2/k_1}$, then we would have $\amm(\optms) = \frac{{k_2}/{k_1}}{\ams(\optmm)} \leq \sqrt{k_2/k_1}$ by Theorem \ref{thm:gen_mmms}. Otherwise, we would have $\ams(\optmm) \leq \sqrt{k_2/k_1}$. Thus, we can conclude that at least one of $\ams(\optmm)$ or $\amm(\optms)$ is always less than or equal to $\sqrt{k_2/k_1}$.
	\end{proof}
	\label{coro:kkmmms}
\end{corollary}
Recall that by Definition \ref{def_ks}, we have that $1 \leq k_1 \leq k, 1 \leq k_2 \leq k$. Therefore, we can see that $\sqrt{k_2/k_1} \leq \sqrt{k/1} = \sqrt{k}$. Hence, we can also conclude the following corollary. 
\begin{corollary}
For any $k$, there always exists a facility location set $A \subseteq \F$ such that choosing $A$ would give a $\sqrt{k}$ approximation both for minimizing $\cmm$ and minimizing $\cms$. In fact, we would either get a $(1, \sqrt{k})$ approximation by choosing $\optmm$ or a $(\sqrt{k},1)$ approximation by choosing $\optms$. In other words, at least one of $\ams(\optmm)$ or $\amm(\optms)$ is always less than or equal to $\sqrt{k}$.
	\label{coro:sqrtkmmms}
\end{corollary}

While Corollary \ref{coro:sqrtkmmms} gives a good upper bound for the simultaneous approximation ratio for $\cmm$ and $\cms$ when $k\in \{1,2,3,4\}$, we can actually show that the upper bound is at most 2 for larger $k$. To show this, we first observe the following lemmas. Let $(i^*,b^*) = \operatorname{argmax}_{i \in \C, b\in \optms} d(i,b)$, $(j^*,a^*) = \operatorname{argmax}_{j \in \C, a\in \optmm} d(j,a)$, $B = \optms\setminus \{b^*\}$.

\begin{lemma}
	For any $j \in \C$, we have that 
	$$(k-2)d(j, b^*) \geq (k-2k_1)d(i^*, b^*).$$
	
	\begin{proof}
	 By triangle inequality, for any $b \in B$, we have that 
	 \[
	 \begin{aligned}
	 	d(i^*,b^*) \leq d(i^*,b) + d(j,b) + d(j,b^*)
	 \end{aligned}
	 \]
	 Therefore, we can see that
	 \[
	 \begin{aligned}
	 	(k-1)d(i^*,b^*) &\leq \sum_{b \in B}d(i^*,b) + \sum_{b \in B}d(j,b) + (k-1)d(j,b^*)
	 \end{aligned}
	 \]
	Recall that  $(i^*,b^*) = \operatorname{argmax}_{i \in \C, b\in \optms} d(i,b)$. By Definition \ref{def_ks}, this means that for any $j \in \C$, we have that
	\[
	\begin{aligned}
		\sum_{b \in \optms} d(j, b) \leq \max_{i \in \C} \sum_{b \in \optms} d(i,b) = k_1 \cdot d(i^*,b^*) 
	\end{aligned}
	\]
	Hence, we have that
	\[
	\begin{aligned}
		\sum_{b \in \optms} d(i^*, b) &\leq  k_1 \cdot d(i^*,b^*) \\
		\sum_{b \in B} d(i^*, b) &\leq  (k_1-1) \cdot d(i^*,b^*)
	\end{aligned}
	\]
	Now, with the above results, we have that 
	\[
	\begin{aligned}
		(k-1)d(i^*,b^*) &\leq \sum_{b \in B}d(i^*,b) + \sum_{b \in B}d(j,b) + (k-1)d(j,b^*)\\
		&\leq (k_1-1) d(i^*,b^*) + \sum_{b \in \optms}d(j,b) + (k-2)d(j,b^*)\\
		&\leq (k_1-1) d(i^*,b^*) + k_1\cdot d(i^*,b^*)+ (k-2)d(j,b^*)\\
		&\leq  (2k_1-1) d(i^*,b^*)+ (k-2)d(j,b^*)\\
	\end{aligned}
	\]
	Then, rearrange the above inequality, we have that 
	\[
	\begin{aligned}
		(k-2)d(j,b^*) &\geq (k-2k_1)d(i^*,b^*)
	\end{aligned}
	\]
	as desired. 
	\end{proof}
	\label{lemma:jc}
\end{lemma}

\begin{lemma}
	For any $a \notin \optms$, there must exist some $j$ such that $d(j,b^*) \leq d(j, a)$.
	\begin{proof}
		Let $A = \{a\} \cup B$. First, note that since $A$ is not the optimal solution (or, it is also an optimal solution but $A \neq \optms$) for $\cms$, there must exist some $j \in \C$ such that 
		$$\sum_{p \in \optms}d(j, p) \leq \sum_{p \in A}d(j, p)$$
	To see why this is true, assume otherwise, $\sum_{p \in \optms}d(j, p) > \sum_{p \in A}d(j, p)$ for all $j \in \C$, this means that $\max_{j \in \C}\sum_{p \in \optms}d(j, p) > \max_{j\in \C}\sum_{p \in A}d(j, p)$, but since $\optms$ is the optimal solution for $\cms$, this is a contradiction. Therefore, we can see that 
	\[
	\begin{aligned}
		\sum_{p \in \optms}d(j, p) &\leq \sum_{p \in A}d(j, p)\\
		d(j,b^*)+\sum_{p \in B}d(j, p) &\leq d(j,a) + \sum_{p \in B}d(j, p)\\
		d(j,b^*) &\leq d(j, a)
	\end{aligned}
	\]
	as desired. 
	\end{proof}
	\label{lemma:opt}
\end{lemma}

Finally, we can conclude with the following theorem.
\begin{theorem}
For any $k$, there always exists a facility location set $A \subseteq \F$ such that choosing $A$ would give a 2 approximation both for minimizing $\cmm$ and minimizing $\cms$. In fact, we would either get a $(1, 2)$ approximation by choosing $\optmm$ or a $(2,1)$ approximation by choosing $\optms$. In other words, at least one of $\amm(\optms)$ or $\ams(\optmm)$ is always less than or equal to $2$.
	\begin{proof}
		We claim that $\min(\amm(\optms), \ams(\optmm)) \leq 2$. To show this, recall that by Corollary \ref{coro:kkmmms} and Definition \ref{def_ks}, we have that $\min(\amm(\optms), \ams(\optmm)) \leq \sqrt{k_2/k_1} \leq \sqrt{k/k_1}$. Then, we consider two cases. First, consider the case where $k - 2k_1 \leq 0, k \leq 2k_1$, this means that $\min(\amm(\optms), \ams(\optmm)) \leq \sqrt{k/k_1} \leq \sqrt{2k_1/k_1} = \sqrt{2} < 2$. Next, we consider the case where $k - 2k_1 > 0, k > 2k_1$. Assume $\optmm \neq \optms$, otherwise choosing $\optms$ would be a 1 approximation for both objectives. This means that there exists some $a \in \optmm, a \notin \optms$. Here note that by Lemma \ref{lemma:opt} and Lemma \ref{lemma:jc}, there exist some $j \in \C$ such that
		\[
		\begin{aligned}
			d(i^*, b^*) &\leq \frac{k-2}{k-2k_1}d(j,b^*)\\
			&\leq \frac{k-2}{k-2k_1}d(j,a)\\
			&\leq \frac{k-2}{k-2k_1}d(j^*,a^*)\\
		\end{aligned}
		\]
		Now, divide both side by $d(j^*,a^*)$, we have that $\amm(\optms) \leq \frac{k-2}{k-2k_1}$,
		$ \min(\amm(\optms), \ams(\optmm)) \leq$ $ \min\left(\sqrt{\frac{k}{k_1}}, \frac{k-2}{k-2k_1}\right )$. Let $k_1\cdot x = k > 2k_1, x > 2$. We then have  
		\[
		\begin{aligned}
			\min(\amm(\optms), \ams(\optmm)) &\leq \min\left(\sqrt{\frac{k_1\cdot x}{k_1}}, \frac{k_1 \cdot x-2}{k_1 \cdot x-2k_1}\right)\\
			&= \min\left(\sqrt{x}, \frac{k_1 \cdot x-2}{k_1(x-2)}\right)\\
			&= \min\left(\sqrt{x}, \frac{x}{x-2}- \frac{2}{k_1(x-2)}\right)\\
			&\leq \min\left(\sqrt{x}, \frac{x}{x-2}\right)
		\end{aligned}
		\]
	The last inequality holds since $k_1 \geq 1, x >2$. We then consider two cases. First, consider the case $x > 4$, this means that $\frac{x}{x-2} < 2$. Next, consider the case where $x \leq 4$, we have that $\sqrt{x} \leq 2$. Therefore, we can conclude that $\min(\amm(\optms), \ams(\optmm)) \leq \min\left(\sqrt{x}, \frac{x}{x-2}\right) \leq 2$. Finally, since both cases give the result $\min(\amm(\optms), \ams(\optmm)) \leq 2$, we can see that at least one of $\amm(\optms)$ or $\ams(\optmm)$  is always less than or equal to $2$ as desired.
	\end{proof}
	\label{thm:2mmms}
\end{theorem}
Combining Corollary \ref{coro:sqrtkmmms} and Theorem \ref{thm:2mmms} together, we can conclude that there always exists a solution that is a $\min(\sqrt{k},2)$ approximation for both $\cmm$ and $\cms$. In fact, we can show that this bound is tight when $k=2$. 

\begin{theorem}
    When $k=2$, no deterministic algorithm can approximate both $\cmm$ and $\cms$ simultaneously within a factor of better than $\sqrt{2}$.
\label{thm:max_boundmmms}
\end{theorem}
\begin{proof}
    Consider the case where there are four locations on a triangle as shown in Figure \ref{fig:ex4}. 
    There is one client on location $C$, one client on location $D$, two facility locations at $A$, one facility location at $C$ and one facility location at $D$ (i.e. the multiset of facility locations is $\F = \{A,A,C,D\}$). We also have that $d(A,C) = d(A,D) = \sqrt{2}/2$, $d(C,D) =  1$. Note that in this case $\optmm$ is taking both facilities on $A$, $\optms$ is taking facilities on $C$ and $D$. W.l.o.g, we also consider solution $H$ which chooses a facility on location $A$ and a facility $C$. We obtain the values as shown in Table \ref{tab:ex4}.

    \begin{figure}[t]
  \begin{center}
  \begin{tikzpicture}[
            roundnode/.style={rectangle, draw=blue!60, fill=green!5, very thick, minimum size=5mm},
            squarednode/.style={rectangle, draw=blue!60, fill=blue!5, very thick, minimum size=5mm},
            scale = 1, transform shape
        ]
        \node[squarednode] (A) at (0,3) {$A$};
        \node[roundnode] (C) at (-5,0) {\textit{C}};
        \node[roundnode] (D) at (5,0) {\textit{D}};
        
        \draw (A) -- (C) node[midway, left=10pt] {$\frac{\sqrt{2}}{2}$};
        \draw (A) -- (D) node[midway, right=10pt] {$\frac{\sqrt{2}}{2}$};
        \draw (C) -- (D) node[midway, above] {$1$};
        
        \node[above=10pt] at (A) {2 facilities};
        \node[below=10pt] at (C) {1 client, 1 facility};
        \node[below=10pt] at (D) {1 client, 1 facility};
    \end{tikzpicture}
  \caption{An instance with two client locations C, D and three possible facility locations A, C, D such that $C$ and $D$ can be chosen once but $A$ can be chosen twice, with distances between any two adjacent locations as well as the number of clients and/or facilities at each location labeled.}
  \label{fig:ex4}
  \end{center}
\end{figure}
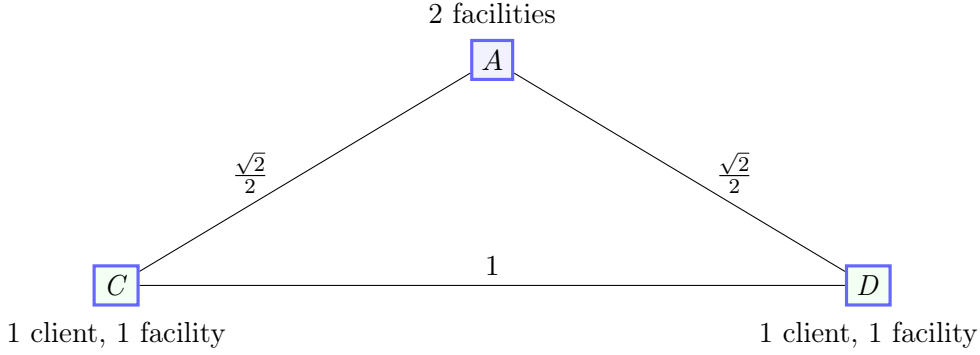

\begin{table}[]
\centering
\begin{tabular}{|c|c|c|c|}
\hline
       & $\optms$ & $\optmm$             & $H$                      \\ \hline
$\cms$ & $1$      & $\sqrt{2}$           & $1 + \frac{\sqrt{2}}{2}$ \\ \hline
$\cmm$ & $1$      & $\frac{\sqrt{2}}{2}$ & $1$                      \\ \hline
$\ams$ & 1                                 & {$\sqrt{2}$} & {$1 + \frac{\sqrt{2}}{2}$} \\ \hline
$\amm$ & {$\sqrt{2}$} & 1                                 & {$\sqrt{2}$}               \\ \hline
\end{tabular}
\caption{An instance for Theorem \ref{thm:max_boundmmms},with $\cmm$, $\cms$,  $\amm$, $\ams$ listed for $\optmm$, $\optms$ and $H$.}
\label{tab:ex4}
\end{table}
    
    Note that in this case, the only options of locations are both facilities on $A$ ($\optmm$), facilities on $C,D$ ($\optms$) and one facility on $A$ and one on $C$ or $D$ ($H$). By choosing $\optmm$, $\optms$ gives us a $\sqrt{2}$ for approximating both objectives and choosing $H$ would lead to a worse result. Therefore, there does not exist a solution that would give an approximation ratio better than $\sqrt{2}$. 
\end{proof}

\subsection{\textsc{Sum-Sum} and \textsc{Sum-Max}}
Lastly, we look at \textsc{Sum-Sum} and \textsc{Sum-Max}. As we have have discussed in Section \ref{section:mmms}, note that when $k=1$, both objectives become the same (minisum) and previous results from Section \ref{section:prelim} do not apply to this set of objectives. Let $\optss$ be the optimal solution for $\css$, $\optsm$ be the optimal solution for $\csm$. Similar to Theorem \ref{thm:gen_mmms}, we can obtain the following result.

\begin{theorem}
	For any $k$, given the optimal facility set $\optsm$ that minimizes $\csm$ and the optimal facility set $\optss$ that minimizes $\css$, we have that $\ass(\optsm) \cdot \asm(\optss) = k$. 
	
	\begin{proof}
		Similar to the proof of Theorem \ref{thm:gen_mmms}, we can show that
		\[
		\begin{aligned}
			\sum_{i \in \C} \max_{a \in \optss} d(i,a) &\leq \sum_{i \in \C} \sum_{a \in \optss} d(i,a)\\
			&= \frac{1}{\ass(\optsm)} \sum_{i \in \C} \sum_{a \in \optsm} d(i,a)\\
			&\leq \frac{k}{\ass(\optsm)} \sum_{i \in \C} \max_{a \in \optsm} d(i,a)
		\end{aligned}
		\]
	Then, divide both sides by $\sum_{i \in \C} \max_{a \in \optsm} d(i,a)$, we have that 
	\[
	\begin{aligned}
		\asm(\optss) &= \frac{k}{\ass(\optsm)}\\
		\ass(\optsm) \cdot \asm(\optss) &= k
	\end{aligned}
	\]
	as desired.
	\end{proof}
	\label{thm:gen_sssm}
\end{theorem}

Then, using the same proof as we have shown in Corollary \ref{coro:kkmmms}, we can see that 

\begin{corollary}
	For any $k$, there always exists a facility location set $A \subseteq \F$ such that choosing $A$ would give a $\sqrt{k}$ approximation both for minimizing $\csm$ and minimizing $\css$. In fact, we would either get a $(1, \sqrt{k})$ approximation by choosing $\optsm$ or a $(\sqrt{k},1)$ approximation by choosing $\optss$. In other words, at least one of $\ass(\optsm)$ or $\asm(\optss)$ is always less than or equal to $\sqrt{k}$.

	\label{coro:ksssm}
\end{corollary}

Recall that Theorem \ref{thm:3sssm} gives a 3 approximation for approximating both \textsc{Sum-Sum} and \textsc{Sum-Max} simultaneously. Now, combined with Corollary \ref{coro:ksssm}, we can conclude that there always exists a solution that is a $\min(\sqrt{k},3)$ approximation for both $\optsm$ and $\optss$. 

\section{Conclusion}
We studied the metric facility location problem  (or, equivalently, the committee selection problem), and showed that various natural objectives are compatible with each other, meaning that we can always find a solution which is close-to-optimal for both of these objectives. We did this by first generalizing some existing work to placing multiple facilities, but our main contributions are the more precise and intricate analysis of how $\css$ and $\cms$ can be simultaneously approximated when placing multiple facilities. Moreover, we provided new approaches for analysing objectives with different individual costs for the clients, such as $\cmm$ and $\cms$, and proved that such objectives can both be approximated within a small constant factor; no previous work has approximated these objectives together to the best of our knowledge.

Despite our results, many interesting open questions remain. While some of our bounds are tight and thus the best possible, it may still be possible to improve some of our other bounds, as shown in Figure \ref{fig:results}. We believe that many of our results and techniques can be greatly generalized to other objective functions; this is the subject of our current work. Finally, although $\css$ and $\cmm$ can be computed in polynomial time, it would be very interesting to see how well multiple objectives in our facility location setting can be approximated {\em efficiently}, as most of our current results only show existence of a good solution, and not its efficient computation.

\subsubsection*{Acknowledgments} This work was partially supported by National Science Foundation award CCF-2006286.

\bibliography{refs}

\appendix

\section*{Appendix}

\section{Extension to the $l$-$centrum$ Problem}

As we have briefly mentioned in Section \ref{section:prelim}, the extension we showed in that section not only hold for $\cm$ and $\cs$ as defined in Definition \ref{def:maxsum}, but also for all $l$-$centrum$ problems. This is due to the fact that \cite{han2022optimizing} considers the $l$-$centrum$ problem while $\cm$ and $\cs$ are merely special cases of such family of problems. We define the $l$-$centrum$ problem as follows.

\begin{definition}
	Let $A \subseteq \F, |A| = k, 1 \leq l \leq n$. We order the $n$ clients $v^1, v^2, \cdots, v^n$ such that 
	$$f\left(v^1, A\right) \geq f\left(v^2, A\right) \geq \ldots \geq f\left(v^n, A\right)$$
	Where $f:M \times M^{k} \rightarrow \mathbb{R^{+}}\cup \{0\}$. Then, we define $l$-$centrum$ ($c_l$) as
	$$c_l(A) = \sum_{i = 1}^{l} f(v^i, A).$$
	
	\label{def:cent}
\end{definition}

\section{Polynomial-time 3-Distortion of Multi-Winner Voting}
\label{section:appen}


As we have discussed in the Introduction, the facility location problems can also be applied in a spacial voting setting where candidates and voters are located in some metric space $(M,d)$. Under this setting, in this section, we will consider the case where we don't know the voter's exact locations but their ordinal preferences for the candidates, namely, {\it distortion} (see the survey \cite{anshelevich2021distortion}). We define {\it distortion} as follows.

\begin{definition}
Let $(\mathcal{M}, \mathcal{D})$ be the set of all metric spaces that are consistent with the voters' ordinal preferences for the candidates that are given, and $O_{d'}$ be the optimal $k$-candidate set for metric $(M', d')$. Denote the cost function by $c$, then the {\em distortion} of a feasible solution $A$ is
    \[
    \sup_{d' \in \mathcal{D}} \frac{ c(A,d')}{c(O_{d'},d')}.
    \]
    \label{def:distortion}
\end{definition}

Now, if we are given a set of $n$ voters $\C$ and a set of candidates $\F$ in metric space $(M,d)$, the goal is to choose $k$ candidates that forms a committee $A \subseteq \F$ such that it minimizes some objectives. Instead of considering $\cm$ and $\cs$ as we did in Section \ref{section:prelim}, we would instead consider a generalized version of those two objectives named the $l$-$centrum$ problem as defined in Definition \ref{def:cent}. As before, $f$ can be any function that satisfies the triangle Inequality \ref{equ:f}, including max and sum functions.

Note that with Definition \ref{def:cent}, $\cm$ is equivalent to $c_1$ and  $\cs$ is equivalent to $c_n$. \cite{caragiannis2022metric} presented a $9$-distortion polynomial-time multi-winner voting rule for $f$ being the $q$-$social$ $cost$ of a committee of size $k$ with $q > k/2$ with $c_n$ in Corollary 3 of their paper. They have also shown that $q$-$social$ $cost$ with $q > k/2$ obeys Inequality \ref{equ:f} in Lemma 2 of their paper, and thus satisfies the desirable property for our function $f$. Because of this, with the mechanism proposed in \cite{Kizilkaya2022PluralityVA}, we can improve the distortion for their objective function to 3, as shown in the following more general theorem. This theorem provides a mechanism for multi-winner voting with distortion at most 3. It works for any objective function where the function $f$ determining the cost of a committee for a single voter obeys the triangle inequality (e.g., max, sum, or $q$-social cost with $q>k/2$), and the objective function combining the costs of individual voters is any $l$-centrum objectives $c_l$. 

\begin{theorem}
    For any $1 \leq l \leq n$, consider $l$-centrum objective $c_l$ with any $f$ such that 
    $$f(v,A) \leq f(v,B) + f(p,B) + f(p,A) \qquad \forall v,p\in \C, A \subseteq \F, B \subseteq \F.$$
    Then there exists a polynomial-time multi-winner voting rule for $c_l$ with distortion at most 3.
\begin{proof}
    We will utilize Plurality Veto presented in \cite{Kizilkaya2022PluralityVA}. The voting role is as follows. For each voter $i$, let $A_{i}$ be the top $k$-committee for $i$ with respect to $f$, since $f$ measures how much each voter likes each set of $k$ candidates. Then, run Plurality Veto on a set of new candidates $\F' = \{A_{1}, A_{2}, \cdots, A_n\}$ and voters $\C$ with the same preference profile induced by the original problem, and the winner would be the winner of original problem. 
    
    First note that this would reduce the problem to a voting problem with $n$ voters and $n$ candidates, which means that running Plurality Veto on it would take polynomial-time. To show that this mechanism has distortion at most 3, let $A$ be the winner, $O$ be the optimal solution, and order the clients as $v^1, v^2, \cdots, v^n$ w.r.t. $A$ as defined in Definition \ref{def:cent}. Note that by construction of Plurality Veto, $A$ must be in $\{A_1, A_2, \cdots, A_n\}$. Besides, we know that there exists a matching which maps each voter $v^i$ to some voters $p[i] \in \C$ such that $v^i$ prefers $A$ over $p[i]$'s favorite candidate $A_{p[i]}$ (this is by standard properties of Plurality Veto). Therefore, we have 
    \[
    \begin{aligned}
        c_l(A) &= \sum_{i = 1}^{l} f(v^i, A)\\
        &\leq \sum_{i = 1}^{l} f\left(v^i, A_{p[i]}\right).
    \end{aligned}
    \]

Now, since we have that $f(v,A) \leq f(v,B) + f(p,B) + f(p,A), \forall v,p\in \C, A \subseteq \F, B \subseteq \F$, we have that 
\[
    \begin{aligned}
        c_l(A) 
        &\leq  \sum_{i = 1}^{l} f(v^i, O) +  \sum_{i = 1}^{l} f(p[i], O) +  \sum_{i = 1}^{l} f(p[i], A_{p[i]}).
    \end{aligned}
    \]
Then, since $A_{p[i]}$ is $p[i]$'s favorite candidate, we must have that $\sum_{i = 1}^{l} f(p[i], A_{p[i]}) \leq \sum_{i = 1}^{l} f(p[i], O)$. In addition, since the mapping $p[i]$ forms a matching, and by definition of $c_l$, we have that $\sum_{i = 1}^{l} f(p[i], O)\leq  c_l(O)$ and $\sum_{i = 1}^{l} f(v^i, O) \leq c_l(O)$. Thus, we can then see that 
\[
    \begin{aligned}
        c_l(A)
        &\leq c_l(O) + c_l(O) + c_l(O)\\
        &= 3 \cdot c_l(O)
    \end{aligned}
    \]
Therefore, this voting rule has distortion at most 3, as desired. 
\end{proof}
\label{them:distor3}
\end{theorem}

\end{document}